\documentclass[11pt]{article}

\usepackage{epsf,epsfig,amsmath,amsfonts}
\usepackage{graphicx}

\usepackage{a4wide}
\usepackage{amsmath}
\usepackage{amscd}
\usepackage{amssymb}
\usepackage{amsthm}
\usepackage{amsmath}
\usepackage{latexsym}
\usepackage{bbm}
\usepackage{cases}
\usepackage{graphicx}
\usepackage{subfig}
\usepackage{hyperref}
\usepackage{natbib}
\usepackage{curves}
\usepackage{xcolor}
\usepackage[left=40mm, right=20mm]{geometry}

\newtheorem{thm}{Theorem}[section]
\theoremstyle{plain}
\newtheorem{lem}[thm]{Lemma}
\newtheorem{prop}[thm]{Proposition}

\theoremstyle{definition}
\newtheorem{defi}[thm]{Definition}
\newtheorem{rem}[thm]{Remark}
\newtheorem{assum}[thm]{Assumption}

\newcommand{\supnormP}[1]{\ensuremath{\|{#1}\|_{L^{\infty}(\ktP)}}}

\newcommand{\ktE}{\ensuremath{\mathbb{E}}}
\newcommand{\ktP}{\ensuremath{\mathbb{P}}}
\newcommand{\ktR}{\ensuremath{\mathbb{R}}}

\newcommand{\ktQ}{\ensuremath{\mathbb{Q}}}

\newcommand{\ktF}{\ensuremath{\mathbb{F}}}

\newcommand{\htF}{\ensuremath{\mathcal{F}}}

\newcommand{\Var}{\ensuremath{\mathrm{Var}}}

\newcommand{\e}{\ensuremath{\mathrm{e}}}

\newcommand{\intl}[2]{\ensuremath{\int_{#1}^{#2}}}
\newcommand{\limit}[2]{\ensuremath{\lim_{{#1}\rightarrow{#2}}}}

\numberwithin{equation}{section}

\title
{Optimal liquidation trajectories for the Almgren-Chriss model with
L\'{e}vy processes}

\date{\today}

\author{Arne L\o kka\footnote{
    Department of Mathematics
    Columbia House
    London School of Economics
    Houghton Street, London WC2A 2AE
    United Kingdom
    (a.lokka\@@lse.ac.uk)}
    \and
    Junwei Xu\footnote{
    Department of Mathematics
    Columbia House
    London School of Economics
    Houghton Street, London WC2A 2AE
    United Kingdom
    (j.xu19\@@lse.ac.uk)}}

\begin{document}
\maketitle

\begin{abstract}
We consider an optimal liquidation problem with infinite horizon in the Almgren-Chriss 
framework, where the unaffected asset price follows a L\'{e}vy 
process. The temporary price impact is described by a general function which satisfies
some reasonable conditions. We consider an investor with constant absolute risk aversion,
who wants to maximise the expected utility of the cash received from the sale of his assets,
and show that this problem can be reduced to a deterministic optimisation problem which
we are able to solve explicitly. In order to compare our results with exponential
L\'{e}vy models, which provides a very good statistical fit with observed asset price data for
short time horizons, we derive the (linear) L\'{e}vy process approximation of such models.
In particular we derive expressions for the L\'{e}vy process approximation of the exponential 
Variance-Gamma L\'{e}vy process, and study properties of the corresponding optimal liquidation
strategy. We then provide a comparison of the liquidation trajectories for reasonable parameters
between the L\'{e}vy process model and the classical Almgren-Chriss model. 
In particular, we obtain an explicit expression for the connection between the temporary 
impact function for the L\'{e}vy model and the temporary impact function for the Brownian motion
model (the classical Almgren-Chriss model), for which the optimal liquidation trajectories for the two models coincide.

\bigskip


\end{abstract}


\section{Introduction}

The introduction of electronic trading platforms was followed by an increased interest in
how to split large orders into smaller orders in order to liquidate large asset positions.
An important question for large investors is how to sell a huge number of shares . 
Because of a lack of liquidity in the market it is often not practical to sell all the shares immediately
since this can result in too high an execution cost.
By splitting a large block of orders into 
smaller ones, the investor can often effectively reduce the cost substantially. 
The problem of finding the optimal way to do this has therefore been the subject of  
considerable interest. 

When the investor determines the speed at which to sell the shares, the key components 
are execution cost and market risk. A slow execution speed will result in a low execution cost, but high
market risk. On the other hand, a fast execution speed will result in a low market risk, but high
execution cost.
In most models dealing with optimal execution, Brownian motion is driving the market risk.
However, in reality observed stock price data demonstrate that Brownian motion is not a particularly good 
model for stock prices, especially for shorter time periods. For instance, 
sudden large price movements and the heavy-tailed distribution of log-returns 
can not be captured by Brownian motion. Also, observed logarithmic stock returns 
over short-time horizons are not normally distributed.
On the other hand, 
there has been a lot of theoretical and empirical studies that show that L\'{e}vy processes provide 
a good fit to market data. For detailed discussions, we refer to \citet{MS}, \citet{EK} and \citet{Barn}. 
Because of the reasons explained above, in particular that in practise the time it takes to liquidate often is 
very short and that L\'{e}vy processes provide good statistical models for stock prices over short time
periods, we will in this paper consider models based on L\'{e}vy processes. 

We consider a continuous-time optimal liquidation problem of a single stock in the Almgren-Chriss 
framework with infinite time horizon. The permanent impact function is supposed to be linear, and we 
describe the temporary market impact function in terms of general sufficient conditions 
ensuring that we are able to solve the problem explicitly. 
The unaffected share price is driven by a linear L\'{e}vy process. 
We assume that the large investor is not permitted to buy back shares 
during liquidation, but we can actually show by a dynamic programming argument
that any such strategy would be sub-optimal. The investor is supposed to have 
constant absolute risk aversion (CARA), and the aim is to maximise expected 
utility of the final cash position over a set of admissible liquidation strategies. Following an idea 
introduced in \citet{SST}, the optimisation problem is reduced to an optimisation problem over a set of 
deterministic strategies. 
Moreover, we show that for a general L\'{e}vy process, there is no immediate relationship between 
the optimal strategy for the mean-variance criteria and the optimal strategy for the expected exponential utility,
which holds for the Brownian motion case. We also show 
that when the L\'{e}vy process is a strict submartingale, our problem is ill-posed, and it is 
always optimal to hold on to the shares rather than sell. Then by solving the Hamilton-Jacobi-Bellman 
equation, the optimal liquidation strategy is derived in an explicit form. After that, we provide some 
conditions which determines whether the optimal strategy has a finite termination time. 

The standard way to analyse stock price data is to find the statistics of the log-returns.
This naturally leads to exponential L\'{e}vy models, and most distributions for the driving L\'{e}vy process
in relation to stock price data is of the exponential model type.
Given a specific exponential L\'{e}vy model, we therefore show how to linearise the model
in order to get a model of the form relevant to our paper.
We then provide some examples where we
assume the log-returns of the share price satisfy the variance gamma distribution and where they
satisfy the normal distribution. 
In the variance gamma case we find that the widely used power law market impact function can result in
optimal strategies which liquidate faster than what seems practical. 
We point out that cost from large trading speeds may be underestimated by 
power functions, and that a function with a bigger growth rate may better reflect the cost of execution. 

For an introduction to high-frequency trading and optimal execution, we refer the reader to
\citet{LL}, \citet{aCJP} and \citet{Gu}, but below we provide a brief review of the more
relevant works in connection to this paper.
\citet{BL} introduced a discrete time stock price model with illiquidity 
effects and related problems. Then \citet{AC,AC2} classified the effects in terms of 
permanent and temporary impacts of trading. In this kind of market impact framework, 
various liquidation models were developed. \citet{AC2} introduced a discrete 
time model with linear permanent and temporary impact functions, a deterministic optimal 
trading strategy was derived by mean-variance optimisation. \citet{Almg} generalised 
the model by considering non-linear impact functions. A single-asset continuous model with infinite time 
horizon was introduced in \citet{SS}, a multi-asset finite horizon model was considered by \citet{SST} 
and \citet{Scho} provides a multi-asset infinite horizon model. 
In these papers strategies were derived by maximising expected utilities instead of the
mean-variance criteria.  
\citet{SST} explained the relationship between mean-variance criteria and the expected exponential utility 
criteria in the Almgren-Chriss framework. They also proved that in a finite time horizon, 
when the stock price is driven by a L\'{e}vy process and an investor with exponential utility, the optimal 
strategy is deterministic. \citet{Gath} suggested that instead of dealing with permanent 
and temporary impacts, the market impact should decay over time. Moreover, \citet{OW} introduced 
a limit order book model and calculated the optimal execution strategy for such a model.
Afterwards, several authors considered variations of this limit order book model, 
such as \citet{AS}, \cite{AFS,ASS} and \citet{Lokk}. In the 
literature of continuous models of optimal execution, price processes are often linear and impact 
additive. However, by considering a new optimisation criterion, a model in the Almgren-Chriss 
framework based on geometric Brownian motion with linear market impact was given in \citet{GS}. 
Then, \citet{Schi} extended this model to general square integrable semimartingales. Also, some 
multiplicative impact models are introduced in \citet{FKTW} and \citet{GZ}; in 
particular, \citet{FKTW} demonstrated that the linear model gives an excellent approximation to 
models with prices modelled as a geometric Brownian motion and multiplicative impact in the 
Almgren-Chriss framework.

The structure of this paper is the following. In Section 2 we introduce the model and 
the optimal execution problem. We reduce the problem to a deterministic optimisation problem in Section 3, 
and solve it in Section 4. In Section 5 we show how to linearise exponential L\'{e}vy models,
and illustrate with examples in Section 6. Section 7 contain proofs not covered in the main sections.


\section{Problem formulation}

Let $(\Omega, \htF, \ktP)$ be a complete probability space, equipped with 
a filtration $\ktF=(\htF_t)_{t\geq 0}$ satisfying the usual conditions,
which supports a one dimensional, non-trivial, $\ktF$-adapted L\'{e}vy process $L$. 
We assume that the L\'{e}vy process $L$ possesses the following properties.

\begin{assum}\label{assum1}
$L_1$ has finite second moment. Moreover, the set
$\bigl\{\delta<0\mid \mathbb{E}\bigl[\mathrm{e}^{\delta L_1}\bigr]<\infty\bigr\}$
is non-empty.  
\end{assum}

For future reference, we observe that this assumption ensures that $L_t$ has finite first 
and second moments, for all $t\geq 0$. Hence, $L$ admits the decomposition 
\begin{gather*} 
L_t=\mu t+\sigma W_t+\int_{\ktR}x\,\big(N(t,dx)-t\nu(dx)\big), 
\end{gather*}
where $\mu\in\mathbb{R}$ and $\sigma\geq 0$ are two constants, $W$ is a standard Brownian 
motion, $N$ is a Poisson random measure which is independent of $W$ with compensator $t\nu(dx)$, and $\nu$ is the L\'{e}vy measure associated with $L$ \citep[see e.g.][]{Kypr}.  Set
\begin{gather}\label{bardelta}
\bar{\delta}=\inf\bigl\{\delta<0\mid \mathbb{E}\bigl[\mathrm{e}^{\delta L_1}\bigr]<\infty\bigr\}<0 .
\end{gather}
Then Assumption \ref{assum1} also ensures that the cumulant generating function of $L_1$ is finite on the 
interval $(\bar{\delta},0]$. 

We consider an investor who aims to sell a large amount of shares of a single stock without any time-constraints.
For $t\geq 0$, we denote by $Y_t$ the investor's position in the stock at time $t$, and 
let $y\geq 0$ denote the investor's initial stock position. We consider the following sets
of admissible liquidation strategies.
 
\begin{defi}\label{Def_Y}
Given an initial share position $y\geq0$, 
the set of admissible strategies, denoted by $\mathcal{A}(y)$, consists of all $\ktF$-adapted, 
absolutely continuous, non-increasing processes $Y$ satisfying 
\begin{gather}
\int_0^\infty \Arrowvert Y_t\Arrowvert_{L^\infty(\mathbb{P})}\, dt<\infty\qquad\text{if }\mu\neq 0, 
\label{Yintegrable}
\end{gather}
and
\begin{gather}
\int_0^\infty \Arrowvert Y_t\Arrowvert_{L^\infty(\mathbb{P})}^2\, dt<\infty \qquad\text{if }\mu=0 .
\label{Y^2integrable}
\end{gather}
Let $\mathcal{A}_D(y)$ be the set of all deterministic strategies in $\mathcal{A}(y)$.  
\end{defi}

The reason for operating with different sets of admissibility depending on the drift parameter $\mu$
is related to the asymptotic properties of the cumulant generating function of $L_1$ around 0. 
If $\mu$ is $0$ then the cumulant generating function is of order two around zero, while it is of 
order one if $\mu$ is different from zero (the importance of the cumulant generating function of 
$L_1$ will be explained later). The integrability conditions in (\ref{Yintegrable}) and 
(\ref{Y^2integrable}) make sure that the investor's finial cash position is well-defined (see Proposition \ref{propconverge}), 
and are also necessary in order for the optimisation problem to be well defined (see Remark \ref{Rem_intcdt}). 

Let $Y\in \mathcal{A}(y)$.
Then there exists an $\ktF$-adapted, positive-valued process $\xi$ such that $Y$ admits the representation
\[Y_t=y-\int_0^t\xi_s\,ds,\]
i.e. $-\xi_t$ is the time derivative of $Y$ at time $t$. In the literature of optimal liquidation, the function
$t\mapsto Y_t$ is referred to as the liquidation trajectory and the associated process $\xi$ as the liquidation speed \citep[see][etc]{AC2,Almg}. 


It is common in the optimal liquidation literature to refer to 
the price process observed in the market if the investor does not trade as
the unaffected stock
price process.
Throughout this chapter we assume that the unaffected stock price process is modelled by  
\[s+L_t,\quad t\geq 0,\] 
where $s>0$ is some constant which denotes the initial stock price. In reality, liquidation is often completed in a very short time. 
It is well known that L\'evy processes provide a good fit of the observed stock returns over short time horizons. 
Therefore the model should provide a good balance between the the cost of liquidating the position and the corresponding
market risk.
Following \citet{AC,AC2} and \cite{Almg}, we split market impact
into two components: a permanent impact and a temporary impact. 
We therefore assume that the stock price at time $t\geq 0$ is given by 
\begin{gather}
S_t=s+L_t+\alpha(Y_t-Y_0)-F(\xi_t), \label{linearprice}
\end{gather}
where $\alpha\geq 0$ is a constant describing the permanent impact and $F:[0,\infty)\rightarrow[0,\infty)$ 
is a function describing the temporary impact.
We assume that $F$ satisfies the following assumptions.
\begin{assum} \label{assumF}
The temporary impact function $F:[0,\infty)\rightarrow[0,\infty)$ satisfies that 
\begin{itemize}
\item[(i)]$F\in C([0,\infty))\cap C^1((0,\infty))$; 
\item[(ii)]$F(0)=0$; 
\item[(iii)]the function $x\mapsto xF(x)$ is strictly convex on $[0,\infty)$; 
\item[(iv)]the function $x\mapsto x^2F'(x)$ is strictly increasing,
and it tends to infinity as $x\rightarrow\infty$.
\end{itemize} 
\end{assum}
In the above assumption, condition (iii) serves to ensure convexity of the objective function in the optimisation problem we are going to solve (see (\ref{objfunc})) and hence uniqueness of the solution (see Theorem \ref{thmvarify}); condition (iv) ensures that the value function in our optimisation problem is solved in an explicit form (see Proposition \ref{propsolntoHJB}) and the optimal liquidation speed process can be expressed in a feedback form (see Theorem \ref{thmvarify}). 
Assumption \ref{assumF} 
is satisfied by a large class of functions, for example, $F(x)=\beta x^\gamma$ with $\beta, \gamma>0$ or
\begin{numcases}
{F(x)=}
\beta_1x^{0.6} \qquad\qquad\qquad\qquad\qquad\qquad x\in [0,\bar{x}],\notag\\
\beta_2e^{\gamma(x-\bar{x}+\hat{x})}-\beta_2e^{\gamma\hat{x}}+\beta_1\bar{x}^{0.6} \qquad x\in (\bar{x},\infty), 
\label{Fcombination}
\end{numcases}
where $\beta_1$, $\beta_2$, $\gamma$ and $\bar{x}$ are strictly positive constants and $\hat{x}$ is given by
\begin{gather*}
\hat{x}=\frac{\ln\Bigl(\frac{3\beta_1}{5\beta_2\gamma}\Bigr)-\frac{2}{5}\ln\bar{x}}{\gamma}.
\end{gather*}
Under this assumption, we derive the following technical properties of $F$ for future references. 

\begin{lem}\label{prop3}
$F$ is strictly increasing and  
$\lim_{x\rightarrow 0}xF'(x)=0$. Hence $\lim_{x\rightarrow 0}x^2F'(x)=0$. 
\end{lem}

For $t\geq 0$, let $C^Y_t$ denote the cash position of the investor at time $t$ associated with an admissible strategy $Y$. Denote by $c\in\ktR$ the investor's initial cash position. Then a direct calculation verifies that his cash position at some finite time $T$ is given by
\begin{align}\label{cashatT}
C^Y_T&=c-\int_0^T S_t\, dY_t  \notag \\
   &=c-\bigl(s-\alpha y\bigr)\bigl(Y_T-y\bigr)+\frac{\alpha}{2}\bigl(y^2-Y_T^2\bigr)
      -L_TY_T+\int_0^T Y_{t-}\, dL_t-\int_0^T \xi_tF(\xi_t)\,dt.  
\end{align}
The next result states that the investor's cash position at the end of time is well-defined.

\begin{prop}\label{propconverge}
For any  $Y\in\mathcal{A}(y)$, we have 
\begin{itemize}
\item[(i)]$L_TY_T\rightarrow 0$ in $L^2(\ktP)$, as $T\rightarrow \infty$;
\item[(ii)]$\int_0^\infty Y_{t-}\, dL_t$ is well-defined in $L^1(\ktP)$. 
\end{itemize}
Therefore, 
\begin{gather}
C^Y_\infty=c+sy-\frac{1}{2}\alpha y^2
+\int_0^\infty Y_{t-}\, dL_t-\int_0^\infty \xi_tF(\xi_t)\, dt ,
\qquad a.s. , \label{cashatinf}
\end{gather}
for any $Y\in\mathcal{A}(y)$. 
\end{prop}

From the expression of $C^Y_\infty$, we can make a few 
observations. The term $c+sy$ can be viewed as the initial mark-to-market wealth of the investor. His total loss due to the permanent impact of trading is given by $\frac{1}{2}\alpha y^2$, which is deterministic and only depends on the initial liquidation size. In particular, it does not depend on
the choice of liquidation strategy. The term $\int_0^\infty \xi_tF(\xi_t)\, dt $ represents the total cost 
due to the temporary impact, and it does depend on the liquidation strategy.  The term
 $\int_0^\infty Y_{t-}\, dL_t$ represents the gain or loss due to market volatility. A relatively slow 
 liquidation speed reduces the temporary impact, but provides a substantial market volatility risk. The optimal
 liquidation strategy is therefore a compromise between the loss due to the temporary impact and the market
 volatility risk. We assume that the investor has a constant absolutely risk aversion (CARA), thus his utility function $U$ satisfies $U(x)=-\exp(-Ax)$, 
 for some constant $A>0$. The investor aims to
 maximise the expected utility of his cash position at the end of time, i.e. he wants to solve
\begin{gather}\label{originalprob}
\sup_{Y\in\mathcal{A}(y)}\mathbb{E}\bigl[U\bigl(C_\infty^Y\bigr)\bigr] .
\end{gather}
In view of (\ref{cashatinf}),
this problem takes the form of 
\begin{gather}
\inf_{Y\in\mathcal{A}(y)}\mathrm{e}^{-A\widetilde{C}}\,
\ktE\biggl[\exp\biggl(-\int_0^\infty AY_{t-}\, dL_t
+A\int_0^\infty \xi_tF(\xi_t)\, dt\biggr)\biggr],  \label{key}
\end{gather}
where
\[\widetilde{C}=c+sy-\frac{1}{2}\alpha y^2.\] 
To solve the above problem, it is sufficient to look at 
\begin{gather}
\inf_{Y\in\mathcal{A}(y)}\ktE\biggl[\exp\biggl(-\int_0^\infty AY_{t-}\, dL_t
+A\int_0^\infty \xi_tF(\xi_t)\, dt\biggr)\biggr].  \label{optprob1}
\end{gather}

\begin{rem}\label{Rem_intcdt}
Suppose that we do not impose integrability conditions (\ref{Yintegrable}) and (\ref{Y^2integrable}) on 
an admissible strategy. The cash position at time infinity may then not be well-defined, but one may 
consider to solve the problem 
\begin{gather*}
\sup_{Y\in\mathcal{A}(y)}\ktE\Bigl[-\exp\Bigl(-A\limsup_{T\rightarrow\infty}C^Y_T\Bigr)\Bigr]. 
\end{gather*}
However, without (\ref{Yintegrable}) and (\ref{Y^2integrable}), our model admits an arbitrage in some week sense. To see this, take
for instance the L\'evy process $L$ to be a standard Brownian motion and consider some stock price $p>s$. Write 
$\tau_p=\inf\{t\geq 0\,|\,L_t\geq p\}$ which is finite a.s. (see \cite{RW}, Lemma 3.6). Suppose $Y$ is an absolutely continuous, non-increasing 
strategy which consists of waiting until time $\tau_p$ and then decreases to 0 in a deterministic way during a finite time, 
i.e. $(Y_{\tau_p+t})_{t\geq 0}$ is a deterministic process starting from $y$. Such strategy is admissible. Let $\xi$ be the associated speed process. We calculate that 
\begin{align*}
&\sup_{Y\in\mathcal{A}(y)}\ktE\Bigl[-\exp\Bigl(-A\limsup_{T\rightarrow\infty}C^{Y}_T\Bigr)\Bigr]\cr
\geq&\,\ktE\Bigl[-\exp\Bigl(-A\limsup_{T\rightarrow\infty}C^{Y}_T\Bigr)\Bigr]\cr
\geq&\,\ktE\Bigl[-\exp\Bigl(-AC^{Y}_{T+\tau_p}\Bigr)\Bigr]\cr
=&\,-\exp\biggl(-A\widetilde{C}+A\int_0^{T} \xi_{t+\tau_p}F\bigl(\xi_{t+\tau_p}\bigr)\,dt\biggr)
\ktE\biggl[\exp\biggl(-\int_0^{T+\tau_p} AY_{t}\,dW_t\biggr)\biggr]\cr
=&\,-\exp\biggl(-A\widetilde{C}+A\int_0^{T} \xi^{p}_{t+\tau_p}F\bigl(\xi^{p}_{t+\tau_p}\bigr)\,dt\biggr)
\ktE\biggl[\exp\biggl(-AyW_{\tau_p}-\int_{\tau_p}^{T+\tau_p} AY_{t}^p\,dW_t\biggr)\biggr]\cr
=&\,-\exp\biggl(-Ayp-A\widetilde{C}+A\int_0^{T} \xi^{p}_{t+\tau_p}F\bigl(\xi_{t+\tau_p}\bigr)\,dt\biggr)
\ktE\biggl[\exp\biggl(\int_{\tau_p}^{T+\tau_p} \frac{1}{2} A^2\bigl(Y_{t}\bigr)^2\,dt\biggr)\biggr]\cr
=&\,-\exp\biggl(-Ayp-A\widetilde{C}+A\int_0^{T} \xi_{t+\tau_p}F\bigl(\xi_{t+\tau_p}\bigr)\,dt
+\int_{0}^{T} \frac{1}{2} A^2\bigl(Y_{t+\tau_p}\bigr)^2\,dt\biggr), 
\end{align*}
where $\widetilde{C}=c+sy-\frac{1}{2}\alpha y^2$, and notice that the two integrals in the above line are two constants. Taking $p$ to $+\infty$ gives 
\begin{gather*}
\lim_{p\rightarrow\infty}\ktE\Bigl[-\exp\Bigl(-AC^{Y}_{T+\tau_p}\Bigr)\Bigr]=0 ,
\end{gather*}
and hence that the associated value function is degenerate.
Moreover, Jensen's inequality results in 
\begin{gather*}
\lim_{p\rightarrow\infty}-\exp\Bigl(-A\ktE\bigl[C^{Y}_{T+\tau_p}\bigr]\Bigr)
\geq\lim_{p\rightarrow\infty}\ktE\Bigl[-\exp\Bigl(-AC^{Y}_{T+\tau_p}\Bigr)\Bigr]=0, 
\end{gather*}
which implies that 
\begin{gather*}
\lim_{p\rightarrow\infty}\ktE\bigl[C^{Y}_{T+\tau_p}\bigr]=\infty .
\end{gather*}
However, $Y$ clearly violates (\ref{Yintegrable}) and (\ref{Y^2integrable}). This shows that 
(\ref{Yintegrable}) and (\ref{Y^2integrable}) are not only convenient from a mathematical point of view, but also necessary in order for
the problem to be well formulated.
\end{rem}

\section{Problem simplification}

Throughout this section, we reduce problem (\ref{optprob1}) to a deterministic optimisation problem. 
Set $\bar{\delta}_A=-\bar{\delta}/A$, where $\bar{\delta}$ is the negative number appearing
in (\ref{bardelta}) and $A$ is the risk aversion parameter appearing in the utility function $U$. We make the following
futher assumptions. 

\begin{assum}\label{assum2}
The initial stock position $y$ is strictly less than $\bar{\delta}_A$. 
\end{assum}

\begin{assum}\label{mu-assump}
The drift $\mu$ of the L\'{e}vy process $L$ satisfies $\mu\leq 0$. 
\end{assum}

Assumption \ref{assum2} puts restrictions on the size of the investor's initial position in order to ensure
that the objective function is finite and well defined. If we do not impose this restriction, then the market risk associated
with the investors position is so large that the investor would want to reduce the position immediately at any finite costs, which is
not possible. 
Assumption \ref{mu-assump} excludes a degenerate case of our reduced problem (see the discussion after equation (\ref{J(Y)finite})). 

Define a function $\kappa_A:[0,\bar{\delta}_A)\rightarrow\ktR$ by $\kappa_A(x)=\kappa(-Ax)$,
where $\kappa$ is the cumulant generating function of $L_1$, that is
\begin{gather*}
\kappa(x)=\ln\bigl(\mathbb{E}\bigl[\mathrm{e}^{xL_1}\bigr]\bigr) ,\qquad x\in\mathbb{R} .
\end{gather*} 
This function will play an important role in the sequel. 

\begin{lem}
The function $\kappa_A$ possesses the following properties 
\begin{itemize}
\item[(i)]$\kappa_A(0)=0$;
\item[(ii)]$\kappa_A$ is strictly convex; 
\item[(iii)]if $\mu=0$, then $\lim_{x\rightarrow 0}\frac{\kappa_A(x)}{x^2}=K$, for some constant $K>0$; 
\item[(iv)]if $\mu\neq 0$, then $\lim_{x\rightarrow 0}\frac{\kappa_A(x)}{x}=-A\mu$. 
             \label{propkappaprop}
\end{itemize}
\end{lem}


\begin{lem}\label{corr2}
Let $Y$ be a continuous process starting form 
$y\in[0,\bar{\delta}_A)$. Then 
\[\int_0^\infty \Arrowvert Y_t\Arrowvert_{L^\infty(\mathbb{P})}^i\, dt<\infty\] 
if and only if 
\[\intl{0}{\infty}\kappa_A\bigl(\Arrowvert Y_u\Arrowvert_{L^\infty(\mathbb{P})}\bigr)\,du<\infty,\] 
where $i=1$ if $\mu<0$, and $i=2$ if $\mu=0$. Moreover, with $\mu>0$, 
\[\int_0^\infty \Arrowvert Y_t\Arrowvert_{L^\infty(\mathbb{P})}\, dt<\infty\]
implies
\[\intl{0}{\infty}\kappa_A\bigl(\Arrowvert Y_u\Arrowvert_{L^\infty(\mathbb{P})}\bigr)\,du<\infty.\] 
\end{lem}


In order to reduce problem (\ref{optprob1}), we also require the following technical result.

\begin{lem}\label{propui}
For any $Y\in\mathcal{A}(y)$, the process $M^Y$ given by 
\begin{gather} \label{M}
M^Y_t=\exp\biggl(\intl{0}{t}-AY_{u-}\,dL_u-
\intl{0}{t}\kappa_A(Y_{u})\,du\biggr),\qquad t\geq 0,
\end{gather}
is a uniformly integrable martingale.  
\end{lem}

It follows from Lemma \ref{corr2} and Lemma \ref{propui} that, for any $Y\in\mathcal{A}(y)$, 
the process $M^Y$ is a strictly positive martingale closed by $M^Y_\infty $. We can therefore 
define a new probability measure 
$\ktQ^Y$ by 
\[\frac {d\ktQ^Y }{d\ktP }=M^Y_\infty.\] 

Based on the idea in \citet{SST} Theorem 2.8, 
and with reference to (\ref{optprob1}) and Lemma \ref{propui}, we calculate that 
\begin{align}
&\inf_{Y\in\mathcal{A}(y)}\ktE\biggl[\exp\biggl(-\int_0^\infty AY_{t-}\, dL_t
+A\int_0^\infty \xi_tF(\xi_t)\, dt\biggr)\biggr]\notag\\
=&\inf_{Y\in\mathcal{A}(y)}\ktE\biggl[\exp\biggl(-\int_0^\infty AY_{t-}\, dL_t
-\intl{0}{\infty}\kappa_A(Y_t)dt+\intl{0}{\infty}\Bigl(\kappa_A(Y_t)
+A\xi_tF(\xi_t)\Bigr)dt\biggr)\biggr]\notag\\
=&\inf_{Y\in\mathcal{A}(y)}\ktE^{\ktQ^Y}\biggl[\exp\biggl(\intl{0}{\infty}\Big(\kappa_A(Y_t)+
A\xi_tF(\xi_t)\Big)\, dt\biggr)\biggr]\notag\\
\leq &\inf_{Y\in\mathcal{A}_D(y)}\exp\biggl[
\int_0^\infty \biggl(\kappa_A(Y_t)+A\xi_tF(\xi_t)\biggr)\,dt\,\biggr] . \label{reducedtoAD}
\end{align}
Now suppose that $Y^*$ is a solution to problem 
\begin{gather*}
\inf_{Y\in\mathcal{A}_D(y)}\exp\biggl[
\int_0^\infty \biggl(\kappa_A(Y_t)+A\xi_tF(\xi_t)\biggr)\,dt\,\biggr] .
\end{gather*}
since $\mathcal{A}_D(y)\subset\mathcal{A}(y)$.
Then it must also be a solution to problem (\ref{optprob1}), and hence equality holds in (\ref{reducedtoAD}). 
This is because otherwise there must be some $\tilde{Y}\in\mathcal{A}_D(y)$ which coincides with some 
sample path of some $Y\in\mathcal{A}(y)$ such that
\begin{align*}
&\exp\biggl[
\int_0^\infty \biggl(\kappa_A(\tilde{Y}_t)+A\xi_tF(\tilde{\xi}_t)\biggr)\,dt\,\biggr]\notag\\
< &\ktE^{\ktQ^Y}\biggl[\exp\biggl(\intl{0}{\infty}\Big(\kappa_A(Y_t)+
A\xi_tF(\xi_t)\Big)\, dt\biggr)\biggr]\notag\\
< &\exp\biggl[
\int_0^\infty \biggl(\kappa_A(Y^*_t)+A\xi_tF(\xi^*_t)\biggr)\,dt\,\biggr]. 
\end{align*}
This contradicts $Y^*$ being a solution to problem (\ref{reducedtoAD}). 
We conclude that it is sufficient to solve the problem 
\begin{gather}
V(y)=\inf_{Y\in\mathcal{A}_D(y)}J(Y), \quad y\in[0,\bar{\delta}_A) \label{valuefunc}
\end{gather}
where $V$ denotes the value function and $J$ is given by 
\begin{gather}\label{objfunc}
J(Y)=\intl{0}{\infty}\biggl(\kappa_A(Y_t)+A\xi_tF(\xi_t)\biggr)\, dt. 
\end{gather}

If we take $Y\in\mathcal{A}_D(y)$ such that $Y_t=\bigl(t-\sqrt{y}\bigr)^2$, for $t\in[0,\sqrt{y}]$, 
and $Y_t=0$, for $t>\sqrt{y}$, then it can be checked that 
\begin{gather}\label{J(Y)finite}
J(Y)=\intl{0}{\sqrt{y}}\biggl(\kappa_A\Bigl(\bigl(t-\sqrt{y}\bigr)^2\Bigr)
+A\bigl(2\sqrt{y}-2t\bigr)F\bigl(2\sqrt{y}-2t\bigr)\biggr)\, dt<\infty, 
\end{gather}
which implies that $V<\infty$. Lemma \ref{propkappaprop} implies $\kappa_A\geq 0$, if $\mu\leq 0$. Hence we have $0\leq V<\infty$, for all $\mu\leq 0$. 
%

Assumption \ref{mu-assump} excludes some degeneracy. 
To see this, suppose $\mu>0$. Then Lemma \ref{propkappaprop} (iv) 
implies that there exists some constant $z>0$ such that $-\infty<\kappa_A(z)<0$. 
Suppose that the investor's initial stock position is $z$ and consider the strategy 
$Y\in\mathcal{A}_D(z)$ satisfying
$Y'_t=-\xi_t=0$ for $t\in[0,s]$ with some $s>0$. Then 
\begin{gather*}
V(z)\leq\int_0^s\kappa_A(z)\,dt+V(z)=s\kappa_A(z)+V(z).
\end{gather*}
This can happen only if $V(z)=-\infty$. 
Let $\bar{Y}\in\mathcal{A}_D(y)$ with $y\geq z$ and set $t_z=\inf\{t\geq 0\mid \bar{Y}_t=z\}<\infty$. Then
\[V(y)\leq\int_0^{t_z}\biggl(\kappa_A(\bar{Y}_t)+A\bar{\xi}_tF(\bar{\xi}_t)\biggr)\,dt+V(z),\]
which implies that $V(y)=-\infty$.  As $z$ can be chosen to be
arbitrarily close to zero, it follows that $V(y)=-\infty$, for all $y\in(0,\bar{\delta}_A)$. 
We therefore conclude that the value function is degenerate when $\mu>0$.
Let $y\in(0,\bar{\delta}_A)$, and suppose (in order to get a contradiction) that there exists an optimal 
strategy $Y^*\in\mathcal{A}_D(y)$. Define $\tilde{\kappa}_A$ to be the function which is identical to $\kappa_A$ with $\mu=0$. 
Then with reference to the L\'{e}vy-Khintchine representation of $L$ (see (\ref{L-K})), we have $\kappa_A(x)=-A\mu x+\tilde{\kappa}_A(x)$. 
By Assumption \ref{assumF} and Lemma \ref{propkappaprop}, we have 
that $\tilde{\kappa}_A(Y^*_t)+A\xi_tF(\xi^*_t)$ is positive. 
Thus, 
\[V(y)=\intl{0}{\infty}\biggl(-A\mu Y^*_t+\tilde{\kappa}_A(Y^*_t)+A\xi^*_tF(\xi^*_t)\biggr)\, dt=-\infty, \quad \mu>0,\]
implies $\int_0^\infty Y^*_t\, dt=\infty$, which contradicts the definition of an admissible strategy.
We conclude that if $\mu>0$,  then there is no optimal admissible liquidation strategy. 





\begin{rem}\label{Rem_meanvariance}
It is mentioned in \citet{SST} that for the Almgren-Chriss model with Brownian motion describing the  
unaffected stock price, the problem of optimising the final cost/reward for a CARA investor over 
a set of adapted strategies provides the same optimal solution as for the problem of optimising for the same
model over deterministic strategies, but with a mean-variance optimisation criterion. 
When the unaffected stock price is not a Brownian motion, but a general L\'{e}vy process,
this relationship no longer holds.
To see this, we know that for our optimisation problem,  
the set of admissible strategies $\mathcal{A}(y)$ can be replaced by $\mathcal{A}_D(y)$. Then in view of (\ref{key}), 
it suffices to consider 
\[\inf_{Y\in\mathcal{A}_D(y)}\mathbb{E}\bigl[\e^{-AC^Y_\infty}\bigr],\]
where 
\[C^Y_\infty=c+sy-\frac{1}{2}\alpha y^2+\int_0^\infty Y_{t-}\, dL_t-\int_0^\infty \xi_tF(\xi_t)\, dt.\]
It can be calculated that 
\begin{gather*}
\mathbb{E}[C^Y_\infty]=c+sy-\frac{1}{2}\alpha y^2+\mu\int_0^\infty  Y_t\, dt-\int_0^\infty \xi_tF(\xi_t)\, dt 
\end{gather*}
and 
\begin{gather*}
\Var(C^Y_\infty)=\sigma^2\intl{0}{\infty}Y_t^2\,dt
+\intl{0}{\infty}\biggl(\int_{\ktR}Y_t^2x^2\,\nu(dx)\biggr)\,dt.
\end{gather*}
Then, 
\begin{align*}
&\,\mathbb{E}\bigl[\exp\bigl(-AC^Y_\infty\bigr)\bigr]\notag\\
          =&\,\exp\biggl[-A\mathbb{E}[C^Y_\infty]+\frac{1}{2}A^2\sigma^2\intl{0}{\infty}Y_t^2\,dt+
              \intl{0}{\infty}\int_{\ktR}
              \Bigl(e^{-AY_tx}-1+AY_tx\Bigr)\,\nu(dx)\,dt\,\biggr]\notag\\
          =&\,\exp\biggl[-A\mathbb{E}[C^Y_\infty]+\frac{1}{2}A^2\Var(C^Y_\infty)
              +\intl{0}{\infty}\int_{\ktR}
              \Bigl(e^{-AY_tx}-1+AY_tx-\frac{1}{2}A^2Y_t^2x^2\Bigr)\,\nu(dx)\,dt\,\biggr]. \label{mean-var}
\end{align*}
From the above expression, it is clear that the problem is equivalent to
\[\sup_{Y\in\mathcal{A}_D(y)}\mathbb{E}[C^Y_\infty]-\frac{1}{2}A\Var(C^Y_\infty), \]
if $\nu(\ktR)\equiv 0$, i.e. the L\'{e}vy process $L$ has no jumps.  
 However, for any general L\'{e}vy process, this equivalence does not hold. 
\qed
\end{rem}

\begin{rem} \label{Rem_buy}
Suppose that the investor is also allowed to buy 
shares. Then in order for the final cash position to be well-defined, we need, in addition to the conditions in Definition \ref{Def_Y}, 
to assume that any admissible strategy $Y$ satisfy $\limit{t}{\infty}t\supnormP{Y_t}=0$ (see Lemma \ref{lemtZ} and proof of Proposition \ref{propconverge} for more details). We also suppose $Y$ is non-negative, that $Y_t<\bar{\delta}_A$ for all $t\geq 0$, and that 
$Y_t=y+\int_0^t\xi_u\,du$ with $\xi_t\in\ktR$. Denote by $\mathcal{A}^{\pm}(y)$ the set of all such admissible strategies, 
and by $\mathcal{A}^{\pm}_D(y)$ the collection of all deterministic admissible strategies. Then by similar arguments as previously,
the liquidation problem can be reduced to
\[V(y)=\inf_{Y\in\mathcal{A}^{\pm}_D(y)}\intl{0}{\infty}\biggl(\kappa_A(Y_t)+A|\xi_t|F\bigl(|\xi_t|\bigr)\biggr)\, dt.\]
Let $Y\in\mathcal{A}^{\pm}_D(y)$ be a strategy including intermediate buying. Then there exist 
times $r$ and $s$ with $r<s$ such that $Y_r=Y_s$ and $Y_t>Y_r$ for all $t\in(r,s)$. 
Consider an admissible strategy $X$ such that $X_t=Y_r$ for $t\in(r,s)$ and $X_t=Y_t$ for 
$t\in[0,r]\cup[s,\infty)$. Then with reference to Lemma \ref{propkappaprop}, 
\begin{gather*}
\int_r^s\biggl(\kappa_A(X_u)+A|\xi^X_u|F\bigl(|\xi^X_u|\bigr)\biggr)\,du =\kappa_A(X_r)(s-r)<\int_r^s\biggl(\kappa_A(Y_u)+A|\xi_u|F\bigl(|\xi_u|\bigr)\biggr)\,du, 
\end{gather*}
where $\xi^X$ is the speed process associated with $X$. Therefore, $J(X)<J(Y)$. This shows $Y$ is not optimal. 
So even if it is allowed, it is not optimal to do any intermediate buying of shares.
%
%
\qed
\end{rem}

\section{Solution to the problem}

With reference to the previous section, recall that the original optimal liquidation problem (\ref{originalprob}) 
is equivalent to solving  
\begin{gather*}
V(y)=\inf_{Y\in\mathcal{A}_D(y)}\intl{0}{\infty}\biggl(\kappa_A(Y_t)+A\xi_tF(\xi_t)\biggr)\, dt ,
\end{gather*}
with
\[dY_t=-\xi_t\, dt, \qquad Y_0=y\in[0,\bar{\delta}_A).\]
According to the theory of optimal control, the corresponding 
Hamilton-Jacobi-Bellman equation is given by
\begin{gather}\label{HJB}
\kappa_A(y)+\inf_{x\geq 0}\bigl\{Ax F(x)-xv'(y)\bigr\}=0 ,  
\end{gather}
with associated boundary condition $v(0)=0$. Let $G:[0,\infty)\rightarrow[0,\infty)$ denote the inverse 
function of $x\mapsto x^2F'(x)$. Assumption \ref{assumF} and Lemma \ref{prop3} together imply that 
$G$ is a continuous, strictly increasing function satisfying $G(0)=0$. 

\begin{prop} \label{propsolntoHJB}
Equation (\ref{HJB}) with boundary condition $v(0)=0$ has a classical solution given by
\begin{gather}
v(y)=\int_{0}^{y}\biggl\{\frac{\kappa_A(u)}{G\bigl(\frac{\kappa_A(u)}{A}\bigr)}
+AF\biggl(G\biggl(\frac{\kappa_A(u)}{A}\biggr)\bigg)\biggr\}\,du,
\qquad 0\leq y<\bar{\delta}_A .  \label{v(y)}
\end{gather}
\end{prop}
The next result provides an expression for the optimal liquidation strategy, and states that the value function $V$ 
identifies with the function $v$ in (\ref{v(y)}). 
\begin{thm}\label{thmvarify}
Let $y\in[0,\bar{\delta}_A)$. Define 
\begin{gather}
\tau=\int_0^{y}\frac{1}{G\bigl(\frac{\kappa_A(u)}{A}\bigr)}\,du.  \label{tau}
\end{gather} 
Let $Y^*$ satisfy 
\begin{gather} \label{opt_strategy}
\intl{Y^*_t}{y}\frac{1}{G\bigl(\frac{\kappa_A(u)}{A}\bigr)}\,du=t,\,\,\text{ if }t\leq\tau,
\quad\text{ and }\quad Y^*_t=0,\,\,\text{ if }t>\tau. 
\end{gather}
Then $Y^*\in\mathcal{A}_D(y)$, and its associated speed process $\xi^*$ satisfies 
\begin{gather}
\xi^*_t=G\biggl(\frac{\kappa_A(Y^*_t)}{A}\biggr), \qquad\text{ for all }t\geq 0.  \label{xi*}
\end{gather}
Moreover, $V$ in (\ref{valuefunc}) is equal to $v$ in (\ref{v(y)}), 
for all $y\in[0,\bar{\delta}_A)$, and $Y^*$ is the unique optimal liquidation strategy for problem (\ref{originalprob}). 
\end{thm}

Note that since is $G$ continuous, (\ref{xi*}) implies that the strategy $Y^*$ 
in (\ref{opt_strategy}) is continuously differentiable. Since the functions $\kappa_A$ and $G$ are both strictly 
increasing, it follows from (\ref{xi*}) that with a larger stock position at time $t$, the associated 
optimal liquidation speed at time $t$ is larger. Moreover, it can be shown by the strict convexity of the cumulant 
generating function of $L_1$ that $A\mapsto\kappa_A(x)/A$ is strictly increasing. Hence, the optimal 
liquidation speed at any time is strictly increasing in the risk aversion parameter $A$. These two 
relations coincide with the intuition that with a larger position of shares, the investor potentially encounters 
bigger risk from the market volatility, as any tiny fluctuation of the stock price will be amplified by the large number of 
shares held. It is therefore optimal to liquidate faster.  Also if the investor is more risk averse, then he 
cares more about the volatility risk, which makes him employ a liquidation strategy with a larger 
speed of sale. Observe that given an initial stock position $y\in[0,\bar{\delta}_A)$, the quantity $\tau$ in (\ref{tau}) 
indicates the liquidation time for the optimal liquidation strategy $Y^*$. 
Depending on the properties of the temporary impact function $F$, $\tau$ may or may 
not be finite.
The next result provides some sufficient conditions for  
the optimal liquidation strategy $Y^*$ to have a finite liquidation time.

\begin{prop}
Under the condition that $y>0$
\begin{itemize}
\item[(i)]suppose $\mu<0$ and there exist constants $p<1$ and $K>0$ such that 
$\lim_{x\rightarrow 0}x^pF'(x)=K$, then $\tau<\infty$.
\item[(ii)]suppose $\mu=0$ and there exist constants $p<1$ and $K>0$ such that 
$\lim_{x\rightarrow 0}x^pF'(x)=K$. If $p\in[0,1)$, then $\tau=\infty$. If $p<0$, then $\tau<\infty$. 
\end{itemize}  \label{proptau}
\end{prop}

\section{Approximation for exponential L\'evy model}

In models for stock prices involving L\'evy processes, it is common to consider exponential L\'evy processes \citep[see e.g.][etc]{MS,EK,Barn}. 
However, it is common in the optimal liquidation literature to use linear model as opposed to exponential models due to tractability and
the short time horizons involved. For practical implementation of our model one could of course directly fit the data to a linear L\'{e}vy model.
However families of distributions that fit observed stock market data well for the exponential L\'{e}vy model are known and obviously the
distribution of the jumps change when you take the exponential. We therefore investigate how to linearise exponential L\'{e}vy models
and how this affects the L\'{e}vy measure.
To this end, we are going to derive a L\'{e}vy process which can be regarded as a linear approximation for a corresponding exponential L\'{e}vy process. We show that this 
L\'{e}vy process satisfies all of the assumptions of being the driving process of the unaffected stock price in the 
liquidation model introduced in previous sections. Therefore, our optimal liquidation strategy derived in the previous section can be 
regarded as an approximation for the result of the corresponding exponential L\'{e}vy model. This linear approximation argument is reasonable since (the majority of) liquidation is usually completed in a very short time. 

Consider a non-trivial, one dimensional, $\ktF$-adapted L\'{e}vy process $\tilde{L}$ 
which admits the canonical decomposition
\begin{gather} \label{Lbar}
\tilde{L}_t=\tilde{\mu}t+\tilde{\sigma}\tilde{W}_t+\int_{|z|\geq 1}z\,\tilde{N}(t,dz)+\int_{|z|<1}z\,\big(\tilde{N}(t,dz)-t\tilde{\nu}(dz)\big), 
\quad t\geq 0, 
\end{gather}
where $\tilde{\mu}\in\ktR$ and $\tilde{\sigma}\geq 0$ are two constants, $\tilde{W}$ is a standard Brownian motion, 
$\tilde{N}$ is a Poisson random measure which is independent of $\tilde{W}$ with compensator $t\tilde{\nu}(dz)$, and $\tilde{\nu}$ is the L\'{e}vy measure associated with $\tilde{L}$. 
We assume that $\tilde{L}$ possesses the following properties. 

\begin{assum} \label{assum3}
We assume that $\tilde{\nu}$ is absolutely continuous with respect to Lebesgue measure, and that 
\begin{gather}
\int_{|z|\geq 1}\mathrm{e}^{2z}\,\tilde{\nu}(dz)<\infty.  \label{finiteexpomoment}
\end{gather}
\end{assum}

Suppose the unaffected stock price is described by the process $\tilde{S}^u$ satisfying 
\begin{gather*} 
\tilde{S}^u_t=\tilde{s}\exp\bigl(\tilde{L}_t\bigr), \quad t\geq 0,  
\end{gather*}
where $\tilde{s}>0$ is some constant denoting the initial stock price. 
Note that (\ref{finiteexpomoment}) 
ensures $\tilde{S}_t^u$ is be square integrable, for all $t\geq 0$ \citep[see e.g.][Theorem 3.6]{Kypr}.
Suppose the affected stock price at time $t\geq 0$ is given by 
\begin{gather*}
\tilde{S}_t=\tilde{s}\exp\bigl(\tilde{L}_t\bigr)+I_t, 
\end{gather*}
where $I_t=\alpha(Y_t-Y_0)-F(\xi_t)$ is the price impact at time $t$ appearing in the previous 
liquidation model with function $F$ satisfying Assumption \ref{assumF} \citep[][study a liquidation model with the affected stock price in this form with a geometric Brownian motion]{GS}. 
By It\^{o}'s formula, for all $t\geq 0$, $\tilde{S}_t$ can be rewritten as 
\begin{gather*}
\tilde{S}_t=\tilde{s}+\int_0^t\tilde{S}_{u-}^u\tilde{m}\,du+\int_0^t\tilde{S}_{u-}^u\tilde{\sigma}\,d\tilde{W}_u
+\int_0^t\int_{\ktR}\tilde{S}_{u-}^u\bigl(\mathrm{e}^z-1\bigr)\,\big(\tilde{N}(t,dz)-t\tilde{\nu}(dz)\big)+I_t,
\end{gather*}
where $\tilde{m}=\tilde{\mu}+\frac{\tilde{\sigma}^2}{2}
+\int_{\ktR}(\mathrm{e}^z-1-z\mathbbm{1}_{\{|z|<1\}})\,\tilde{\nu}(dz)$. 
In order to approximate the exponential L\'evy model, consider the process $\hat{S}$ such that 
\begin{gather*} 
\hat{S}_t=\tilde{s}+\tilde{s}\tilde{m}t+\tilde{s}\tilde{\sigma}\tilde{W}_t
+\int_{\ktR}\tilde{s}\bigl(\mathrm{e}^z-1\bigr)\,\big(\tilde{N}(t,dz)-t\tilde{\nu}(dz)\big)+I_t, \quad t\geq 0, 
\end{gather*}
which can be considered as a linear approximation of $\tilde{S}$. Recall that the affected stock price in the preceding model is given by 
\begin{gather*}
S_t=s+L_t+I_t, \quad t\geq 0, 
\end{gather*}
where $L_t=\mu t+\sigma W_t+\int_{\ktR}x\,\big(N(t,dx)-t\nu(dx)\big)$. Comparing this to the expression of 
$\hat{S}_t$, it can be seen that if we take $s=\tilde{s}$ and choose $L$ to be such that 
\begin{gather} \label{candidate}
L_t=\tilde{s}\tilde{m}t+\tilde{s}\tilde{\sigma}\tilde{W}_t
+\int_{\ktR}\tilde{s}\bigl(\mathrm{e}^z-1\bigr)\,\big(\tilde{N}(t,dz)-t\tilde{\nu}(dz)\big), \qquad t\geq 0, 
\end{gather}
then it follows that 
\[\hat{S}_t=\tilde{s}+L_t+I_t, \quad \text{ for all } t\geq 0. \]
We may therefore consider $\hat{S}$ 
as the affected stock price process in the liquidation model introduced in previous sections. The next proposition verifies that $L$ with the above expression is a L\'{e}vy process satisfying Assumption \ref{assum1}. 

\begin{prop} \label{thmnuL}
Let $L$ be given by (\ref{candidate}). Write $\hat{L}=L/\tilde{s}$. 
Then $\hat{L}$ is an $\ktF$-adapted L\'evy process whose L\'evy measure, denoted by $\hat{\nu}$, satisfies 
\begin{gather*} 
\hat{\nu}(dx)=\frac{1}{x+1}\,\tilde{f}\Bigl(\ln(x+1)\Bigr)\,dx, \qquad x>-1, \, x\neq 0. 
\end{gather*}
Therefore, $L$ is an $\ktF$-adapted L\'evy process satisfying 
Assumption \ref{assum1}.
\end{prop}

\begin{rem}
From equation (\ref{forremarkinsec5}) (in the proof of Proposition \ref{thmnuL}) we know that 
\[\int_{|x|\geq 1}\mathrm{e}^{ux}\,\hat{\nu}(dx)<\infty, \quad\text{ for all }\,u\leq 0.\]
This implies that $\bar{\delta}$ given by (\ref{bardelta}) is equal 
to $+\infty$, and therefore, Assumption \ref{assum2} is satisfied for any initial stock position $y>0$. 
In other words, if we consider an 
exponential L\'evy model and use the approximation scheme discussed above, we do not need to concern any restriction on the maximum volume of liquidation.  \qed
\end{rem}

With $L$ given by (\ref{candidate}) and $\hat{L}$ defined in Proposition \ref{thmnuL}, in view of (\ref{valuefunc})-(\ref{objfunc})
we consider the optimisation problem 
\begin{gather} \label{appro_prob}
V(y)=\inf_{Y\in\mathcal{A}_D(y)}\intl{0}{\infty}
\biggl(\hat{\kappa}_{\tilde{A}}(Y_t)+A\xi_tF(\xi_t)\biggr)\, dt ,\quad y\geq 0, 
\end{gather}
where $A>0$ denotes the investor's risk aversion, $\tilde{A}=A\tilde{s}$ 
and $\hat{\kappa}_{\tilde{A}}:[0,\infty)\rightarrow[0,\infty)$ is defined by 
$\hat{\kappa}_{\tilde{A}}(x)=\hat{\kappa}(-\tilde{A}x)$ with $\hat{\kappa}$ being the cumulant generating 
function of $\hat{L}_1$. 

\begin{thm}\label{Thm_appro}
The unique optimal liquidation speed for problem (\ref{appro_prob}) is given by 
\begin{gather} \label{appro_xi*}
\xi^*_t=G\biggl(\frac{\hat{\kappa}_{\tilde{A}}(Y^*_t)}{A}\biggr), \quad t\geq 0, 
\end{gather}
where $G:[0,\infty)\rightarrow[0,\infty)$ is the inverse function of $x\mapsto x^2F'(x)$ and $Y^*$ is the associated unique optimal admissible stock position process satisfying 
\begin{gather*} 
\intl{Y^*_t}{y}\frac{1}{G\bigl(\frac{\hat{\kappa}_{\tilde{A}}(u)}{A}\bigr)}\,du=t,\,\,\text{ if }t\leq\tau,
\quad\text{ and }\quad Y^*_t=0,\,\,\text{ if }t>\tau, 
\end{gather*}
with $\tau$ defined by
\begin{gather*}
\tau=\int_0^{y}\frac{1}{G\bigl(\frac{\kappa_A(u)}{A}\bigr)}\,du. 
\end{gather*} 
The value function in (\ref{appro_prob}) satisfies 
\begin{gather*}
V(y)=\int_{0}^{y}\biggl\{\frac{\hat{\kappa}_{\tilde{A}}(u)}{G\bigl(\frac{\hat{\kappa}_{\tilde{A}}(u)}{A}\bigr)}
+AF\biggl(G\biggl(\frac{\hat{\kappa}_{\tilde{A}}(u)}{A}\biggr)\bigg)\biggr\}\,du,
\quad y\geq 0 .  
\end{gather*}
\end{thm}

\section{Examples}

In this section, we provide some examples following the approximation scheme discussed in 
the previous section. We consider the process $\tilde{L}$ in (\ref{Lbar}) as a variance gamma (VG) L\'evy process, 
which is obtained by subordinating a Brownian motion using a gamma process. Precisely, we consider 
$\tilde{L}$ to be such that 
\begin{gather*}
\tilde{L}_t=\theta\tau_t+\rho W_{\tau_t}, \quad t\geq 0, 
\end{gather*}
where $\theta\in\ktR$ and $\rho>0$ are some constants, W is a standard Brownian motion and $\tau$ is a 
gamma process such that $\tau_t\sim\Gamma\bigl(\frac{t}{\eta},\frac{1}{\eta}\bigr)$ 
\footnote{$\Gamma(a,b)$ denotes a gamma distribution with shape parameter $a>0$ and rate parameter $b>0$, for 
which the probability density function is given by $f(x)=\frac{b^a}{\Gamma(a)}x^{a-1}\e^{-bx}$, for $x>0$, where 
$\Gamma(\cdot)$ is the gamma function. For any $X\sim\Gamma(a,b)$, $\ktE[X]=\frac{a}{b}$ and $\Var[X]=\frac{a}{b^2}$.}, 
for some constant $\eta>0$. Then $\tilde{L}$ is a VG L\'evy process whose L\'{e}vy density is given by
\begin{gather*}
\tilde{f}(z)=\frac{1}{\eta|z|}\mathrm{e}^{Cz-D|z|}, \quad z\in\ktR, 
\end{gather*}
where 
\[C=\frac{\theta}{\rho^2}\qquad\text{and}\qquad D=\frac{\sqrt{\theta^2+\frac{2\rho^2}{\eta}}}{\rho^2},\]
and its cumulant generating function $\tilde{\kappa}$ admits the expression 
\begin{gather} \label{kappabar}
\tilde{\kappa}(x)=-\frac{1}{\eta}\ln\biggl(1-\frac{x^2\rho^2\eta}{2}-\theta\eta x\biggr) 
\end{gather}
\citep[see e.g.][]{CT}. It can be shown that Assumption \ref{assum3} is satisfied if $D-C>2$. We calculate according to Proposition \ref{thmnuL} that the L\'{e}vy measure $\hat{\nu}$ of the process $\hat{L}$ satisfies 
\begin{numcases}
{\hat{\nu}(dx)=}
\frac{-1}{\eta\ln(x+1)}(x+1)^{C+D-1}\,dx, \qquad x\in (-1,0),\notag\\
\frac{1}{\eta\ln(x+1)}(x+1)^{C-D-1}\,dx, \qquad x\in (0,\infty).\notag  
\end{numcases}
Therefore, the function $\hat{\kappa}_{\tilde{A}}:[0,\infty)\rightarrow[0,\infty)$ in (\ref{appro_prob}), 
denoting it by $\hat{\kappa}_{\tilde{A}}^{VG}$ in the example of VG L\'evy process, is given by 
\begin{gather} \label{kappaofVG}
\hat{\kappa}_{\tilde{A}}^{VG}(u)=-\tilde{A}\tilde{m}u+\int_{-1}^\infty 
\Bigl(\mathrm{e}^{-\tilde{A}ux}-1+\tilde{A}ux\Bigr)\,\hat{\nu}(dx), 
\end{gather}
where the drift parameter $\tilde{m}=\tilde{\kappa}(1)$. 
 
The next result provides 
a lower bound for $\hat{\kappa}_{\tilde{A}}^{VG}$, which later will be useful for 
deciding the limit behaviour of the price impact function. 
\begin{prop} \label{prop_lowerbd}
For $u\geq 0$, write 
\begin{align*}
\underline{\hat{\kappa}}_{\tilde{A}}^{VG}(u)
&=-\tilde{A}\tilde{m}u+\frac{\e}{\eta}\biggl[-\frac{\e^{\tilde{A}u}\tilde{A}u}{C+D+2}
\biggl(\frac{1}{\tilde{A}u}\wedge 1\biggr)^{C+D+2}
+\frac{\e^{\tilde{A}u}}{C+D+1}\biggl(\frac{1}{\tilde{A}u}\wedge 1\biggr)^{C+D+1}\cr
&\qquad+\frac{\tilde{A}u}{C+D+2}-\frac{1+\tilde{A}u}{C+D+1}\biggr], 
\end{align*}
and in particular, for $u\geq \frac{1}{\tilde{A}}$, 
\begin{align*}
\underline{\hat{\kappa}}_{\tilde{A}}^{VG}(u)
&=-\tilde{A}\tilde{m}u+\frac{\e}{\eta}\biggl[\biggl(\frac{1}{\tilde{A}u}\biggr)^{C+D+1}\e^{\tilde{A}u}
\biggl(\frac{1}{C+D+1}-\frac{1}{C+D+2}\biggr)\cr
&\qquad+\frac{\tilde{A}u}{C+D+2}-\frac{1+\tilde{A}u}{C+D+1}\biggr]. 
\end{align*}
Then we have $\hat{\kappa}_{\tilde{A}}^{VG}(u)\geq\underline{\hat{\kappa}}_{\tilde{A}}^{VG}(u)$, for all $u\geq 0$. 
\end{prop}

In order to compare the optimal strategy for the model involving a VG L\'evy process and the optimal strategy
for the corresponding model with a Brownian motion (i.e. when $\tilde{L}$ is a Brownian motion), 
we derive that the function $\hat{\kappa}_{\tilde{A}}:[0,\infty)\rightarrow[0,\infty)$ in 
(\ref{appro_prob}), denoting it by $\hat{\kappa}_{\tilde{A}}^{BM}$, is given by
\begin{gather} \label{kappaofBM}
\hat{\kappa}_{\tilde{A}}^{BM}(u)
=-\tilde{A}\Bigl(\tilde{\mu}+\frac{\tilde{\sigma}^2}{2}\Bigr)u+\frac{1}{2}\tilde{A}^2\tilde{\sigma}^2u^2,
\end{gather}
where $\tilde{\mu}\in\ktR$ and $\tilde{\sigma}>0$ are some constants which represent the drift and 
volatility of $\tilde{L}$, respectively. In the case, Assumption 
\ref{assum3} is always satisfied. 

Throughout this section we use the following parameters for our VG L\'evy 
process: $\theta=-0.002$, $\rho=0.02$ and $\eta=0.6$ (timescale in days).
For more details on empirical studies of parameters of the VG stock price model we refer to 
\citet{RSS}. For parameters in the Brownian motion case, we choose the parameters such that 
the expectation and the second moment of $\e^{\tilde{L}_t}$ match that of the
VG model.
Hence, $\tilde{\mu}$ and $\tilde{\sigma}$ in (\ref{kappaofBM}) are taken to be such that 
$\tilde{\mu}+\frac{\tilde{\sigma}^2}{2}=\tilde{\kappa}(1)$ and $2\tilde{\mu}+2\tilde{\sigma}^2=\tilde{\kappa}(2)$, 
where $\tilde{\kappa}$ is given by (\ref{kappabar}). Therefore, throughout this section, 
\begin{gather*}
\tilde{\mu}=2\tilde{\kappa}(1)-\frac{\tilde{\kappa}(2)}{2}
\qquad\text{ and }\qquad
\tilde{\sigma}^2=\tilde{\kappa}(2)-2\tilde{\kappa}(1). 
\end{gather*}
Moreover, we choose $\tilde{s}=100$ for simplicity.

\subsection{Power-law price impact function}

Consider the power-law temporary impact function, i.e. $F:[0,\infty)\rightarrow[0,\infty)$ is given by
\[F(x)=\beta x^\gamma,\]
where $\beta>0$ and $\gamma>0$ are constants. This kind of impact function is widely believed to be realistic 
and has been well-studied in the literature of price impact \citep[see e.g.][etc]{LFM,ATHL}. It can be checked that $F$ 
satisfies Assumption \ref{assumF}, and the function $G$ appearing in (\ref{appro_xi*}) is given by 
\[G(x)=\biggl(\frac{x}{\beta\gamma}\biggr)^{\frac{1}{\gamma+1}},\qquad x\geq 0.\]
Applying Proposition \ref{proptau}, we see that if $\hat{L}$ is a strict supermartingale, then $\tau$ 
in (\ref{tau}) is finite, for all $\gamma>0$; if $\hat{L}$ is a martingale, then $\tau=\infty$ for 
$\gamma\in(0,1]$, and $\tau<\infty$ when $\gamma>1$. 
It follows from (\ref{appro_xi*}) that the optimal liquidation speed takes the expression 
\begin{gather} \label{optspeed}
\xi^*_t=\biggl(\frac{\hat{\kappa}_{\tilde{A}}(Y^*_t)}{A\beta\gamma}\biggr)^{\frac{1}{\gamma+1}}, 
\qquad\text{ for all }t\geq 0. 
\end{gather}
We adopt the values of $\beta$ and $\gamma$ suggested in \citet{ATHL} where parameters of 
the power-law temporary impact are studied empirically. In particular, we take $\gamma=0.6$ 
and $\beta=4.7\times 10^{-5}$ 
\footnote{With our notations, the temporary impact function $F$ in \citet{ATHL} is given by 
$F(x)=\beta x^\gamma=\tilde{S}_{0-}\tilde{\beta}\tilde{\sigma}\bigl(\frac{x}{\tilde{V}}\bigr)^\gamma$, 
where $\tilde{V}$ denotes the daily volume of a given stock, the value of exponent $\gamma$ is argued 
to be $0.6$ (as the main result in their paper) and $\tilde{\beta}$ is a constant which is suggested 
to be $0.142$. From the values of parameters of the VG L\'evy process that we have chosen, it can be 
calculated that the volatility $\tilde{\sigma}$ in the Brownian motion case is roughly equal to $0.02$. Comparing 
this number to the values of volatilities and daily volumes of stocks provided in examples in 
\citet{ATHL}, we may take $\tilde{V}=2\times 10^{6}$ as a reasonable choice. Moreover, we choose 
$\tilde{s}=100$ for simplicity. Then $\beta$ is calculated to be $4.7\times 10^{-5}$.

Note that the empirical study in \citet{ATHL} is based on a model parametrised by the volume time which is 
defined as fractions of a daily volume. Therefore, any results of number regarding time derived from a model 
with power-law impact function in this section should be interpreted as volume time. }. 

\begin{figure}[htbp]
\centering
\begin{centering}
{\begin{centering}
\hspace*{-0.1cm}\includegraphics[width=5.5in]{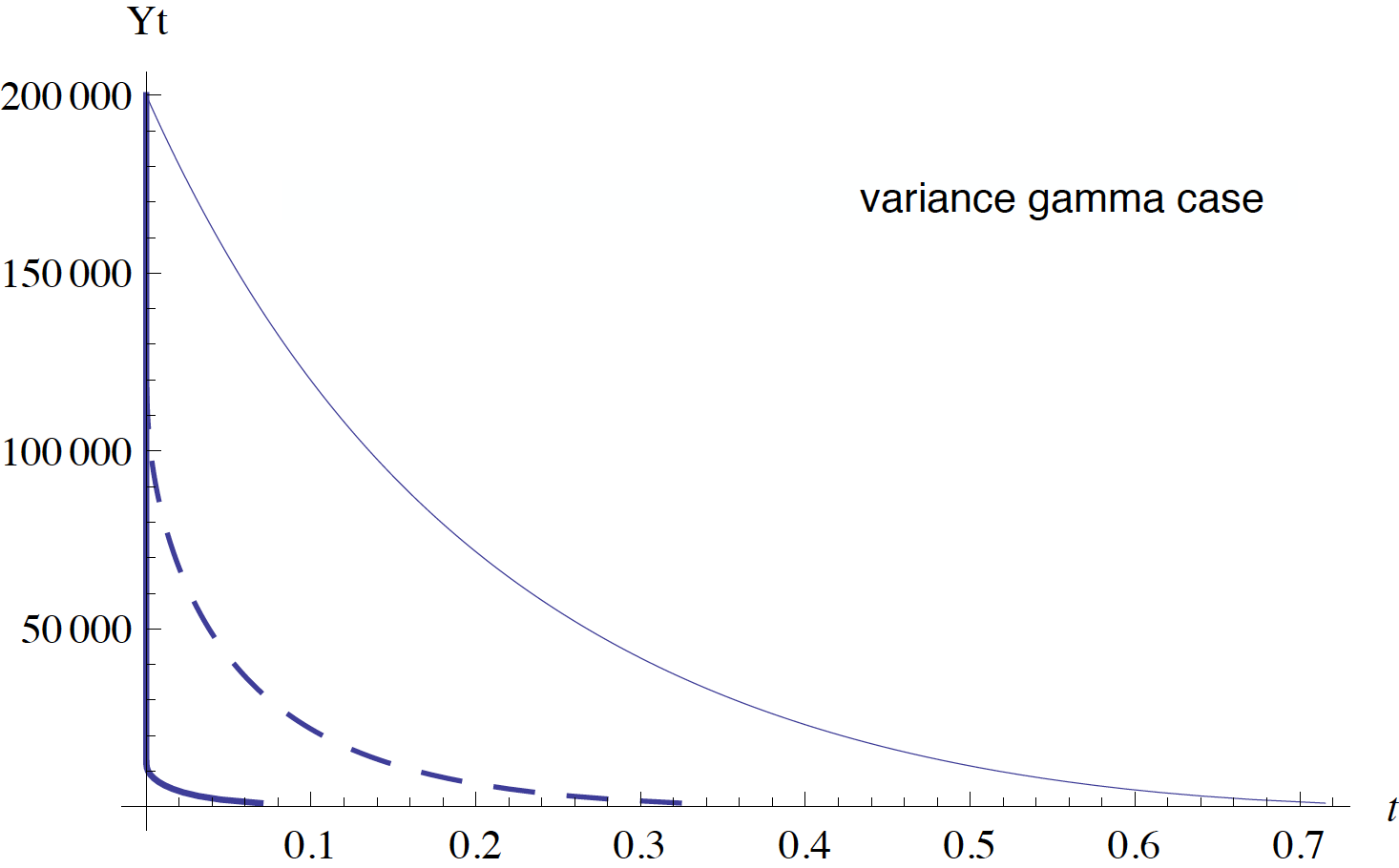}
\end{centering}}
{\begin{centering}
\hspace*{0.14cm}\includegraphics[width=5.5in]{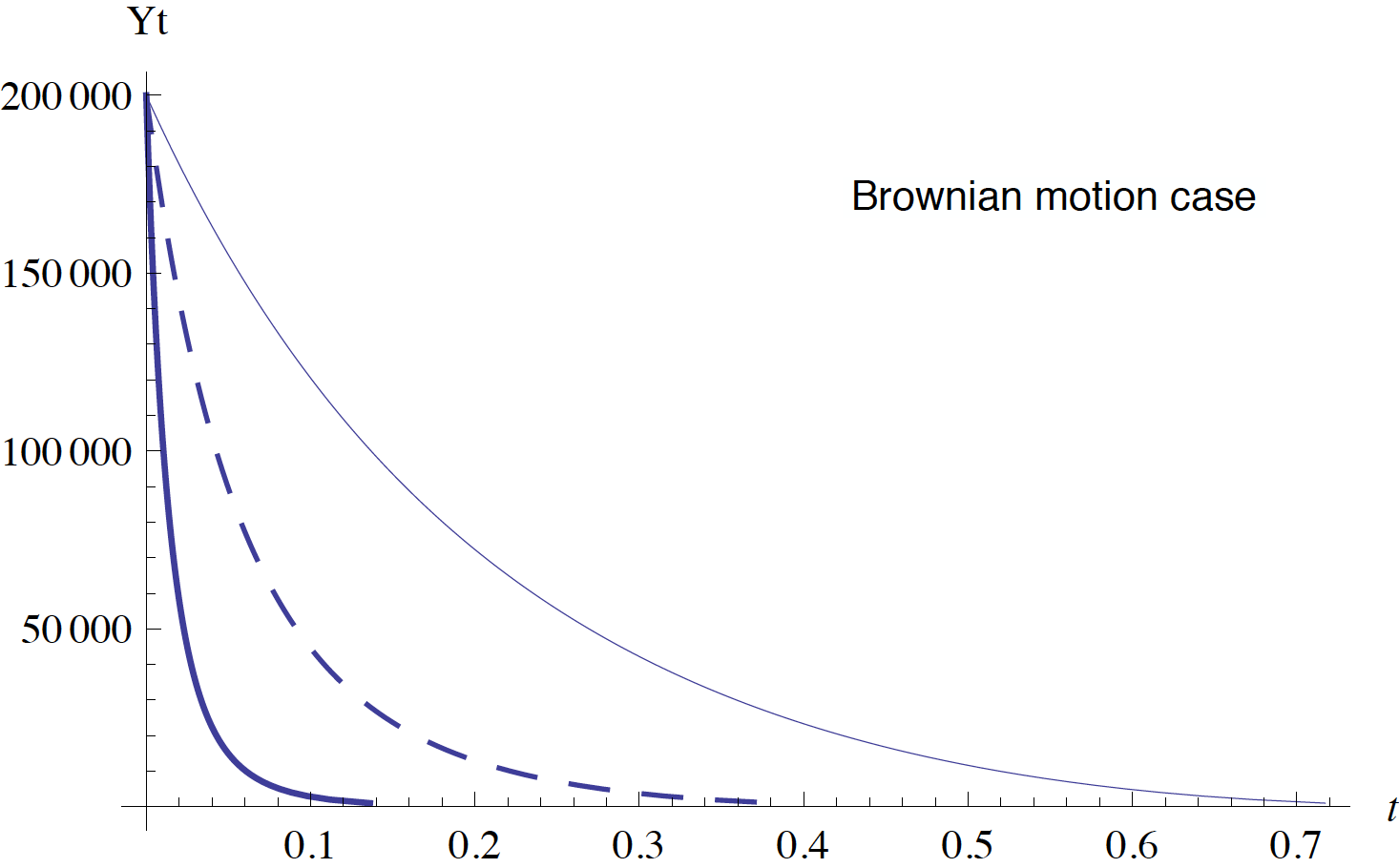}
\end{centering}}
\end{centering}
\caption{{\footnotesize Optimal liquidation trajectories for variance gamma L\'evy process model and Brownian 
motion model with 0.6 power-law temporary impact function. Thin curves are for $A=10^{-6}$, dashed curves are 
when $A=10^{-5}$ and thick curves are for $A=10^{-4}$.}}
\label{power-BMVG}
\end{figure}


Consider a stock with average daily volume $2\times 10^6$. 
Suppose the investor wants to liquidate a position of $2\times 10^5$ 
\footnote{Since the study of parameters of impacts in \citet{ATHL} is based on liquidating the amount of 
shares that weighted as $10\%$ of daily volume, in order to keep consistent with the values of parameters 
of the temporary impact function that we have chosen, we let the initial stock position to be $2\times 10^5$ 
which is $10\%$ of the daily volume that we have chosen as explained before. }
of this stock. Figure \ref{power-BMVG} shows the optimal liquidation trajectories for 
both the VG L\'evy process case and the Brownian 
motion case when the risk aversion parameter $A$ takes values of $10^{-6}$, $10^{-5}$ and $10^{-4}$.
\footnote{In view of the general literature on preferences, these values of $A$ may seem small.
However, in the context of a liquidation model they can viewed as reasonable
if the investor is not sensitive to any large costs which are insignificant compared to 
his total wealth. We refer to \cite{AC2} and \cite{Almg} for more details about the risk aversion parameter for 
the Almgren-Chirss liquidation model.}
We see that when $A=10^{-6}$, the optimal strategies for the two models are almost identical. As $A$ increases, 
the optimal speeds increase in both models, and in particular, speeds increase much faster in the VG model for big positions. 
In each case, liquidation finishes in a short time period, which confirms that the linear approximation 
of the exponential model is reasonable.


As shown in the first graph of Figure \ref{power-BMVG} when $A=10^{-5}$ and $A=10^{-4}$, the 
stock positions drop immediately by a large proportion of its initial value. In order 
to get more details about these two trajectories, we compute that when $A=10^{-5}$ the time spent on 
liquidating $40\%$ of $2\times 10^5$ shares is about $0.00018$ when the  investor follows the optimal 
strategy for the VG case. If the time parametrisation is the same as clock time, then $0.00018$ is just 
a few seconds. If the investor's risk aversion is $A=10^{-4}$, then according to the optimal strategy 
for VG model, he spends roughly $1.34\times 10^{-14}$ amount of time to liquidate 90\% of his initial position.


With a large stock position, due to the nature of jumps of the VG L\'evy process we expect that the investor would liquidate 
much faster than the optimal strategy for Brownian motion model. 
However, the above examples show that with the 0.6 power-law temporary impact function, 
in the VG case, optimal liquidation speeds can be too large for the optimal strategy to be practical, while speeds in the 
Brownian motion model stay in a reasonable range. Intuitively, an unrealistically high optimal liquidation speed can be 
due to price impact for a large trading speed being underestimated. In other words, the cost resulting from large speeds 
is too small. This argument can be confirmed by the expression of the optimal liquidation speed in (\ref{appro_xi*}) 
that if the temporary impact function $F$ has a small growth rate, then growth rate of function $G$ is large, and 
therefore optimal speed can be very high, when stock position is large. 
It is mentioned in \citet{Rosu,Gath} that the impact function should be concave for small trading speeds and convex for 
large speeds. 
Therefore, we next try to explore a mode of growth of the price impact function for which the optimal liquidation 
speeds for the L\'evy model is realistic.


\subsection{A relation between the impact functions of various models}

In this section we derive a connection between a temporary impact function for the L\'evy liquidation model and 
a temporary impact function for the Brownian motion liquidation model
such that the two respective optimal strategies coincide with each other.

Let $F^L:[0,\infty)\rightarrow[0,\infty)$ and $F^{BM}:[0,\infty)\rightarrow[0,\infty)$ be temporary 
impact functions satisfying Assumption \ref{assumF} considered in a L\'evy model and a Brownian motion 
model, respectively. We denote by $G^L:[0,\infty)\rightarrow[0,\infty)$ and $G^{BM}:[0,\infty)\rightarrow[0,\infty)$ 
the inverse functions of $x\mapsto x^2(F^L)'(x)$ and $x\mapsto x^2(F^{BM})'(x)$, respectively. Then 
in view of (\ref{appro_xi*}), the optimal liquidation speed at time $t$ for each model, denoted by $\xi^{L}_t$ 
and $\xi^{BM}_t$, are given by 
\begin{gather*} 
\xi^{L}_t=G^{L}\biggl(\frac{\hat{\kappa}^{L}_{\tilde{A}}(Y^L_t)}{A}\biggr)
\qquad\text{ and }\qquad
\xi^{BM}_t=G^{BM}\biggl(\frac{\hat{\kappa}^{BM}_{\tilde{A}}(Y^{BM}_t)}{A}\biggr), 
\end{gather*}
where $\hat{\kappa}^{L}_{\tilde{A}}$ and $\hat{\kappa}^{BM}_{\tilde{A}}$ are different 
versions for of $\hat{\kappa}_{\tilde{A}}$, and $Y^L$ and 
$Y^{BM}$ are corresponding optimal liquidation strategies in each model. Suppose for all 
$t\geq 0$, $Y^*_t=Y^{L}_t=Y^{BM}_t$, then
\begin{gather} \label{GL=GBM}
G^{L}\biggl(\frac{\hat{\kappa}^{L}_{\tilde{A}}(Y^*_t)}{A}\biggr)
=G^{BM}\biggl(\frac{\hat{\kappa}^{BM}_{\tilde{A}}(Y^{*}_t)}{A}\biggr),\qquad t\geq 0. 
\end{gather}
Write $z=G^{BM}\Bigl(\frac{\hat{\kappa}^{BM}_{\tilde{A}}(Y^{*}_t)}{A}\Bigr)$. So by (\ref{kappaofBM}) we have 
\begin{gather*}
Y^*_t=\frac{\tilde{u}+\sqrt{\tilde{u}^2+2A\tilde{\sigma}^2z^2(F^{BM})'(z)}}{\tilde{A}\tilde{\sigma}^2}, 
\end{gather*}
where $\tilde{u}=\tilde{\mu}+\frac{\tilde{\sigma}^2}{2}$. Then from (\ref{GL=GBM}) we obtain that 
\begin{gather*}
(F^L)'(z)=\frac{1}{Az^2}\,\hat{\kappa}^{L}_{\tilde{A}}\Biggl(\,\frac{\tilde{u}
+\sqrt{\tilde{u}^2+2A\tilde{\sigma}^2z^2(F^{BM})'(z)}}{\tilde{A}\tilde{\sigma}^2}\,\Biggr), 
\end{gather*}
which is equivalent to 
\begin{gather} \label{FL_FBM}
F^L(x)=\int^x_0\frac{1}{Az^2}\,\hat{\kappa}^{L}_{\tilde{A}}\Biggl(\,\frac{\tilde{u}
+\sqrt{\tilde{u}^2+2A\tilde{\sigma}^2z^2(F^{BM})'(z)}}{\tilde{A}\tilde{\sigma}^2}\,\Biggr)\,dz. 
\end{gather}
It can be shown that Assumption \ref{assumF} is satisfied by the above expression. 
We can therefore conclude that if $F^L$ and $F^{BM}$ satisfy (\ref{FL_FBM}), then $Y^L=Y^{BM}$.

Suppose $F^{BM}$ in (\ref{FL_FBM}) follows a power-law such that the optimal speed for the Brownian motion 
case is practically realistic (this kind of model is indeed used in practice). Then it follows from the relation in (\ref{FL_FBM}) 
that in order for the optimal speed in VG case to be practically realistic, the function $F^L$ needs to increase to 
infinity faster than any power function. This is because for the VG L\'evy process case, the lower bound of 
the function $\hat{\kappa}_{\tilde{A}}^{VG}$ given in Proposition \ref{prop_lowerbd} tends to infinity
faster than any power function. 

\section{Proofs}

\begin{proof}[\textbf{Proof of Lemma \ref{prop3}}]
For $\lambda\in(0,1)$ and $x\in(0,\infty)$,
Assumption \ref{assumF} (ii) and (iii) imply that $F(\lambda x)<\lambda F(x)<F(x)$, 
which shows that $F$ is strictly increasing.

The derivative of $x\mapsto xF(x)$, together with the convexity of this function, implies that $\lim_{x\rightarrow 0}xF'(x)$ exists. 
As $F'(x)>0$, for all $x>0$, it follows that $\lim_{x\rightarrow 0}xF'(x)\geq 0$. 
Suppose $\lim_{x\rightarrow 0}xF'(x)>0$.
Then there exist constants $\bar{x}>0$ and $c>0$, such that 
for all $x\in(0,\bar{x})$,  
\[F'(x)>\frac{c}{x}. \]
But then, 
\[F(\bar{x})=\lim_{x\rightarrow 0}\int_x^{\bar{x}}F'(u)\,du
\geq\lim_{x\rightarrow 0}\int_x^{\bar{x}}\frac{c}{u}\,du=\infty,\]
which contradicts the continuity of $F$. Hence, $\lim_{x\rightarrow 0}xF'(x)=0$, and it therefore 
follows that $\lim_{x\rightarrow 0}x^2F'(x)=0$. 
\end{proof}

The next lemma is used in the proof of Proposition \ref{propconverge}. 

\begin{lem}
Let $Z$ be a positive-valued, decreasing process satisfying $\intl{0}{\infty}Z_t dt<\infty$.
Then $tZ_t \rightarrow 0$, as $t\rightarrow \infty$. \label{lemtZ}
\end{lem}

\begin{proof}
Suppose $\liminf_{t\rightarrow \infty}tZ_t>0$, then there exists some constant $c$ such that 
\[\liminf_{t\rightarrow \infty}tZ_t> c>0.\] 
This implies that we can find some $s\geq 0$ such that for all $t\geq s$, 
\[Z_t> \frac{c}{t}.\]
Hence 
\[\intl{s}{\infty}Z_t\, dt\geq \intl{s}{\infty}\frac{c}{t}\,dt=\infty, \]
which contradicts $\int_0^\infty Z_t\, dt<\infty$.
Thus, we have shown that 
\begin{gather}
\liminf_{t\rightarrow \infty}tZ_t=0.  \label{liminf}
\end{gather}
We know that $Z$ is a decreasing process, 
which is of finite variation. By It\^{o}'s formula we calculate that
\[tZ_t=\int_0^t u\, dZ_u+\int_0^t Z_u\, du.\]
It can be observed that $t\mapsto \int_0^t u dZ_u$ 
is negative and decreasing while $t\mapsto \int_0^t Z_u du$ 
is positive and increasing. Then, 
\begin{align}
0&\leq \sup_{t\geq r}tZ_t
\leq\sup_{t\geq r}\int_0^t u\, dZ_u+\sup_{t\geq r}\int_0^t Z_u\, du
=\int_0^r u\, dZ_u+\int_0^\infty Z_u\, du.   \label{suptZ}
\end{align}
Also,
\begin{align}
&\inf_{t\geq r}tZ_t
\geq\inf_{t\geq r}\int_0^t u\, dZ_u+\inf_{t\geq r}\int_0^t Z_u\, du
=\int_0^\infty u\, dZ_u+\int_0^r Z_u\, du.   \label{inftZ}
\end{align}
Taking $r$ to infinity in (\ref{inftZ}) and (\ref{suptZ}), and by (\ref{liminf}) we have 
\begin{align*}
0&\leq\limsup_{t\rightarrow \infty}tZ_t
=\limit{r}{\infty}\sup_{t\geq r}tZ_t
\leq\int_0^\infty u\, dZ_u+\int_0^\infty Z_u\, du,\\
0&=\liminf_{t\rightarrow \infty}tZ_t
=\limit{r}{\infty}\inf_{t\geq r}tZ_t
\geq\int_0^\infty u\, dZ_u+\int_0^\infty Z_u\, du. 
\end{align*}
Therefore, we conclude that $\lim_{t\rightarrow \infty}tZ_t=0$. 
\end{proof}

\begin{proof}[\textbf{Proof of Proposition \ref{propconverge}}]
\mbox{}
\begin{itemize}
\item[(i)]Let $f$ be the characteristic function of $L_t$, so 
\[f(u)=\ktE[\mathrm{e}^{\mathrm{i}uL_t}]= \mathrm{e}^{t\psi(u)},\]
where $\psi(u)$ is given by the L\'{e}vy-Khintchine representation of $L$.
By Assumption \ref{assum1} we know that $f$, hence $\psi$, are twice differentiable at 0. 
Hence, we calculate that $f'(0)=\mathrm{i}\ktE[L_t]=t\psi'(0)$ and $f''(0)=-\ktE[L_t^2]$, and therefore, 
\begin{gather*}
\ktE[L_t^2]=(\mu t)^2-\psi''(0)t. \label{kappafinite}
\end{gather*}
Then, 
\begin{gather}
\ktE\bigl[(L_tY_t)^2\bigr]\leq \ktE[L_t^2]\supnormP{Y_t}^2
=\mu^2 \bigl(t\supnormP{Y_t}\bigr)^2-\psi''(0)t\supnormP{Y_t}^2.  \label{LTYTinq}
\end{gather}
If $\mu\neq 0$, then for any $Y\in\mathcal{A}(y)$, 
$\bigl(\supnormP{Y_t}\bigr)_{t\geq 0}$ and $\bigl(\supnormP{Y_t}^2\bigr)_{t\geq 0}$ are continuous, 
positive and decreasing. The integrability condition in (\ref{Yintegrable}) implies that 
$\int_0^\infty \Arrowvert Y_t\Arrowvert_{L^\infty(\mathbb{P})}^2\, dt<\infty$. Therefore, according to 
Lemma \ref{lemtZ} we have 
\[\limit{t}{\infty}t\supnormP{Y_t}=0\qquad\text{ and }\qquad\limit{t}{\infty}t\supnormP{Y_t}^2=0.\]
Hence, by (\ref{LTYTinq}) and the finiteness of $\mu$ and $\psi''(0)$ we conclude that 
\[\limit{T}{\infty}\ktE\bigl[(L_tY_t)^2\bigr]=0.\]
When $\mu=0$, we get $\int_0^\infty \Arrowvert Y_t\Arrowvert_{L^\infty(\mathbb{P})}^2\, dt<\infty$ directly 
as a condition of admissible strategies. Therefore, the same result follows. 

\item[(ii)]
Using Cauchy-Schwarz inequality and It\^o isometry we obtain 
\begin{align*}
&\ktE\biggl[\,\biggl|\,\intl{0}{T}Y_{t-}\,dL_t\,\biggr|\,\biggr]\cr
\leq\,&|\mu|\ktE\biggl[\,\biggl|\,\intl{0}{T}Y_{t-}\,dt\,\biggr|\,\biggr]
+\ktE\biggl[\,\biggl|\intl{0}{T}Y_{t-}d\biggl(\sigma W_t+\int_{\ktR}x\,\Big(N(t,dx)-t\nu(dx)\Big)\biggr)\biggr|\,\biggr]\cr
\leq\,&|\mu|\intl{0}{T}\|Y_{t}\|_{L^\infty(\ktP)}\,dt
+\ktE\biggl[\,\biggl|\intl{0}{T}Y_{t-}d\biggl(\sigma W_t
+\int_{\ktR}x\,\Big(N(t,dx)-t\nu(dx)\Big)\biggr)\biggr|^2\,\biggr]^{\frac{1}{2}}\cr
=\,&|\mu|\intl{0}{T}\|Y_{t}\|_{L^\infty(\ktP)}\,dt
+\biggl(\sigma^2+\int_{\ktR\backslash\{0\}}x^2\nu(dx)\biggr)^{\frac{1}{2}}
\ktE\biggl[\intl{0}{T}Y_t^2\,dt\biggr]^{\frac{1}{2}}\cr
\leq\,&|\mu|\intl{0}{T}\|Y_{t}\|_{L^\infty(\ktP)}\,dt
+\biggl(\sigma^2+\int_{\ktR\backslash\{0\}}x^2\nu(dx)\biggr)^{\frac{1}{2}}
\ktE\biggl[\intl{0}{T}\|Y_{t}\|_{L^\infty(\ktP)}^2\,dt\biggr]^{\frac{1}{2}}\cr
\end{align*}
From the existence of the first and second moments of $L_1$, we know that $\mu$, $\sigma$ and 
$\int_{\ktR\backslash\{0\}}x^2\nu(dx)$ are all finite. The result then follows from the integrability 
conditions in (\ref{Yintegrable}) and (\ref{Y^2integrable}) of an admissible strategy. 
\end{itemize}
\end{proof}

\begin{proof}[\textbf{Proof of Lemma \ref{propkappaprop}}]
\mbox{}
\begin{itemize}

\item[(i)]Let $\psi(u)$ be given by the L\'{e}vy-Khintchine representation of $L$. 
Then for all $u\in[0,\bar{\delta}_A)$, we have
\begin{align}
\kappa_A(u)=\psi(\mathrm{i}Au)
                  =-A\mu u+\frac{1}{2}A^2u^2\sigma^2
                  +\int_{\ktR}\Bigl(\e^{-Aux}-1+Aux\Bigr)\,\nu(dx). \label{L-K}
\end{align}
Therefore, $\kappa_A(0)=0$ follows directly.

\item[(ii)] Observe that $-A\mu u$, $\frac{1}{2}A^2u^2\sigma^2$ and $\e^{-Aux}-1+Aux$ are all convex in $u$, and 
in particular that $\frac{1}{2}A^2u^2\sigma^2$ and $\e^{-Aux}-1+Aux$ are strictly convex in $u$. Thus, with reference 
to (\ref{L-K}), the strict convexity of $\kappa_A$ can be concluded from the assumption that $L$ is non-trivial.


\item[(iii)]Let $\mu=0$. In view of (\ref{L-K}), in order to proof 
$\lim_{x\rightarrow 0}\frac{\kappa_A(x)}{x^2}=K>0$, it suffices to show that 
\[\lim_{u\rightarrow 0}\int_{\ktR}\biggl(\frac{e^{-Aux}-1+Aux}{A^2u^2}\biggr)\,\nu(dx)=K',\]
for some constant $K'>0$. Let $0<A\bar{u}<\bar{\delta}_A$. It can be checked that for all $u\in(0,\bar{u})$, 
\[\biggl|\frac{e^{-Aux}-1+Aux}{A^2u^2}\biggr|<\frac{x^2}{2},\quad\text{ if }x>0,\]
and 
\[\biggl|\frac{e^{-Aux}-1+Aux}{A^2u^2}\biggr|<\frac{e^{-A\bar{u}x}-1+A\bar{u}x}{A^2\bar{u}^2},\quad\text{ if }x<0.\]
Because of the finite second moment of $L_1$ and the fact that $\kappa_A(\bar{u})<\infty$, both $\frac{x^2}{2}$ and $\frac{e^{-A\bar{u}x}-1+A\bar{u}x}{A^2\bar{u}^2}$ are $\nu$-integrable. 
Thus, by the dominated convergence theorem, it follows that 
\begin{align*}
\lim_{u\rightarrow 0}\int_{\ktR}\biggl(\frac{e^{-Aux}-1+Aux}{A^2u^2}\biggr)\,\nu(dx)
=\int_{\ktR}\frac{x^2}{2}\,\nu(dx)=K',
\end{align*}
where $K'$ is some strictly positive constant.

\item[(iv)]Let $\mu\neq 0$. Then $\lim_{x\rightarrow 0}\frac{\kappa_A(x)}{x}=-A\mu$ follows from (\ref{L-K})
as well as (iii). 
\end{itemize}
\end{proof}

\begin{proof}[\textbf{Proof of Lemma \ref{corr2}}]
Let $\mu=0$. Then Lemma \ref{propkappaprop} (iii) implies that there exists strictly positive constants 
$\bar{x}$, $C_1$ and $C_2$ such that  $C_1x^2<\kappa_A(x)< C_2x^2$,
for all $x\in(0,\bar{x})$.
Suppose that $\int_0^\infty \Arrowvert Y_t\Arrowvert_{L^\infty(\mathbb{P})}^2\, dt<\infty$. 
Then $Y_t$ tends to zero as $t$ tends to infinity. Hence, there exists $s>0$, 
such that $\supnormP{Y_t}\in(0,\bar{x})$, for all $t>s$. Then 
\begin{gather}\label{tmp}
C_1\int_s^\infty\supnormP{Y_t}^2\,dt<\int_s^\infty\kappa_A\bigl(\supnormP{Y_t}\bigr)\,dt
<C_2\int_s^\infty\supnormP{Y_t}^2\,dt,
\end{gather}
from which it follows that 
$\intl{s}{\infty}\kappa_A\bigl(\Arrowvert Y_u\Arrowvert_{L^\infty(\mathbb{P})}\bigr)\,du<\infty$.
Since $\supnormP{Y_t}$ is bounded for $t\in[0,s]$, 
we have $\intl{0}{s}\kappa_A\bigl(\Arrowvert Y_u\Arrowvert_{L^\infty(\mathbb{P})}\bigr)\,du<\infty$. 
A similar argument together with the inequality (\ref{tmp}) also establishes the reverse implication. 
The proofs regarding the cases of $\mu <0$ and $\mu>0$ are  similar to above. 
\end{proof}

\begin{proof}[\textbf{Proof of Lemma \ref{propui}}]
By It\^{o}'s formula and using the expression of $\kappa_A$ in (\ref{L-K}) we calculate that
\begin{align*} 
M^Y_t=&1-\int_0^t M^Y_{u-}AY_{u-}\Bigl(\bigl(\mu-\int_{\ktR}x\,\nu(dx)\bigr)\, du+\sigma\, dW_u\Bigr)\cr
    &\qquad-\int_0^t M^Y_{u-}\Bigl(\kappa_A\bigl(Y_{u-}\bigr)-\frac{1}{2}A^2Y_{u-}^2\sigma^2\Bigr)\,du \notag\\
    &\qquad+\int_0^t\int_{\ktR}M^Y_{u-}\Bigl(e^{-AY_{u-}x}-1\Bigr)
    \,\Bigl(\big(N(du,dx)-\nu(dx)du\big)+\nu(dx)du\Bigr)\cr 
   =&1-\int_0^t M^Y_{u-}AY_{u-}\sigma\, dW_u
    +\int_0^t\int_{\ktR}M_{u-}\Bigl(e^{-AY_{u-}x}-1\Bigr)\,\big(N(du,dx)-\nu(dx)du\big), 
\end{align*}
which shows $M$ is a local martingale. Define 
\begin{align*}
X_t=\intl{0}{t}-AY_{u-}\,d\tilde{L}_u \qquad\text{ and }\qquad
K(\theta)_t=\intl{0}{t}\tilde{\kappa}_A(\theta Y_{u})\,du, 
\end{align*}
where $\theta\in[0,1]$, $Y\in\mathcal{A}(y)$ with $y\in[0,\bar{\delta}_A)$, $\tilde{L}$ is the martingale 
part of $L$ and $\tilde{\kappa}_A$ is equal to $\kappa_A$ with $\mu=0$. It can be checked that 
the process $M^Y$ in (\ref{M}) can be rewritten as 
\begin{gather*}
M^Y=\exp\bigl(X-K(1)\bigr). 
\end{gather*}
With reference to Definition 3.1 and Theorem 3.2 in \citet{KS}, in order to show $M^Y$ is a uniformly 
integrable martingale, it is sufficent to check that 
\begin{gather}
\lim_{\delta\downarrow 0}\sup_{t\in\ktR_+}\delta\log
\Biggl(\ktE\biggl[\exp\biggl(\frac{1}{\delta}\Bigl((1-\delta)K(1)_t
-K(1-\delta)_t\Bigr)\biggr)\bigg]\Biggr)=0 ,  \label{uicdt}
\end{gather}
for $\delta\in(0,1)$.
Observe that 
\begin{align}
&\,\lim_{\delta\downarrow 0}\sup_{t\in\ktR_+}\delta\log
\Biggl(\ktE\biggl[\exp\biggl(\frac{1}{\delta}\Bigl((1-\delta)K(1)_t
-K(1-\delta)_t\Bigr)\biggr)\bigg]\Biggr)\notag\\
   \leq&\,\lim_{\delta\downarrow 0}\sup_{t\in\ktR_+}\delta\log
   \Biggl(\exp\Biggl(\biggl{\|}\frac{1}{\delta}\Bigl((1-\delta)K(1)_t
   -K(1-\delta)_t\Bigr)\biggr{\|}_{L^\infty (\ktP)}\Biggr)\Biggr)\notag\\
   =&\,\lim_{\delta\downarrow 0}\sup_{t\in\ktR_+}\Bigl{\|}(1-\delta)K(1)_t
   -K(1-\delta)_t\Bigr{\|}_{L^\infty (\ktP)}\notag\\
   =&\,\lim_{\delta\downarrow 0}\sup_{t\in\ktR_+}\biggl{\|}(1-\delta)\intl{0}{t}\tilde{\kappa}_A(Y_{u})\,du
      -\intl{0}{t}\tilde{\kappa}_A\bigl((1-\delta)Y_{u}\bigr)\,du\,\biggr{\|}_{L^\infty (\ktP)}\notag\\
   \leq&\,\lim_{\delta\downarrow 0}\sup_{t\in\ktR_+}\intl{0}{t}\Bigl{\|}(1-\delta)\tilde{\kappa}_A(Y_{u})
         -\tilde{\kappa}_A\bigl((1-\delta)Y_{u}\bigr)\Bigr{\|}_{L^\infty (\ktP)}\,du\notag\\
   \leq&\,\lim_{\delta\downarrow 0}\intl{0}{\infty}\Bigl{\|}(1-\delta)\tilde{\kappa}_A(Y_{u})
         -\tilde{\kappa}_A\bigl((1-\delta)Y_{u}\bigr)\Bigr{\|}_{L^\infty (\ktP)}\,du.  \label{useDCT}
\end{align}
For $\delta\in(0,1)$, we have that
\begin{align*}
&\,\bigl{\|}(1-\delta)\tilde{\kappa}_A(Y_{u})-\tilde{\kappa}_A\bigl((1-\delta)Y_{u}\bigr)\bigr{\|}_{L^\infty (\ktP)}\\
\leq&\,\bigl{\|}(1-\delta)\tilde{\kappa}_A(Y_{u})\bigr{\|}_{L^\infty (\ktP)}
     +\bigl{\|}\tilde{\kappa}_A\bigl((1-\delta)Y_{u}\bigr)\bigr{\|}_{L^\infty (\ktP)}\\
   =&\,(1-\delta)\tilde{\kappa}_A\bigl(\supnormP{Y_{u}}\bigr)+\tilde{\kappa}_A\bigl((1-\delta)\supnormP{Y_{u}}\bigr)\\
   \leq&\,2\tilde{\kappa}_A\bigl(\supnormP{Y_{u}}\bigr).
\end{align*}
The last two steps are because $\tilde{\kappa}_A(x)$ is positive and 
non-decreasing for $x\geq 0$, which follow from Lemma  \ref{propkappaprop} (i), (ii) and (iii). 
According to (\ref{Yintegrable}) or (\ref{Y^2integrable}) as well as Lemma \ref{corr2}, we have
\[\intl{0}{\infty}\tilde{\kappa}_A\bigl(\supnormP{Y_t}\bigr)dt<\infty.\]
Then, by the dominated convergence theorem, (\ref{useDCT}) gives
\begin{equation}
\begin{split}
&\,\lim_{\delta\downarrow 0}\sup_{t\in\ktR_+}\delta\log
\Biggl(\ktE\biggl[\exp\biggl(\frac{1}{\delta}\Bigl((1-\delta)K(1)_t
-K(1-\delta)_t\Bigr)\biggr)\bigg]\Biggr)\\
   \leq&\,\intl{0}{\infty}\lim_{\delta\downarrow 0}\Bigl{\|}(1-\delta)\tilde{\kappa}_A(Y_{u})
            -\tilde{\kappa}_A\bigl((1-\delta)Y_{u}\bigr)\Bigr{\|}_{L^\infty (\ktP)}\,du\\
    =&\,0.  \label{LHSleq0}
\end{split}
\end{equation}
On the other hand, the convexity of $\tilde{\kappa}_A(x)$ and $\tilde{\kappa}_A(0)=0$ imply
\[(1-\delta)\tilde{\kappa}_A(x)\geq\tilde{\kappa}_A\bigl((1-\delta)x\bigr),\quad\text{ for }\delta\in(0,1),\]
hence, 
\[(1-\delta)K(1)_t-K(1-\delta)_t\geq 0.\]
Combining this with (\ref{LHSleq0}), we get (\ref{uicdt}). 
\end{proof}

The next lemma is used in the proofs of Proposition \ref{propsolntoHJB} and Theorem \ref{thmvarify}. 

\begin{lem} \label{lem_x/G(x)}
Let the function $F$ satisfy Assumption \ref{assumF}. Then $x\mapsto\frac{x}{G(x)}$ is continuous on 
$[0,\infty)$, where $G:[0,\infty)\rightarrow[0,\infty)$ is the inverse function of $x\mapsto x^2F'(x)$. 
\end{lem}

\begin{proof}
Assumption \ref{assumF} and Lemma \ref{prop3} imply that $G$ is continuous and $G(0)=0$. 
Therefore, it is sufficient to check that $\lim_{x\rightarrow 0}\frac{x}{G(x)}<\infty$. Let $x=u^2F'(u)$. 
Then it follows that $\frac{x}{G(x)}=uF'(u)$. Hence, the result follows from the fact that 
$u\rightarrow 0$, as $x\rightarrow 0$, and $\lim_{u\rightarrow 0}uF'(u)=0$ (see Lemma \ref{prop3}). 
\end{proof}

\begin{proof}[\textbf{Proof of Proposition \ref{propsolntoHJB}}]
We first show that the function $v$ given by (\ref{v(y)})
is continuously differentiable, and note that it is sufficient to show that 
$v'(y)=\frac{\kappa_A(y)}{G\bigl(\frac{\kappa_A(y)}{A}\bigr)}
+AF\bigl(G\bigl(\frac{\kappa_A(y)}{A}\bigr)\big)$ is continuous on $[0,\bar{\delta}_A)$.
This is the case if $x\mapsto\frac{x}{G(x)}$ is continuous for $x\geq 0$. 
But this is demonstrated by Lemma \ref{lem_x/G(x)}. 

Recall that the Hamilton-Jacobi-Bellman equation in our problem is 
\begin{gather*}
\kappa_A(y)+\inf_{x\geq 0}\bigl\{AxF(x)-x v'(y)\bigr\}=0. 
\end{gather*}
In order to prove that $v$ in (\ref{v(y)}) is a solution to this equation, 
because $AxF(x)-xv'(y)$ is strictly convex in $x$, it is enough to show that for all 
$y\in[0,\bar{\delta}_A)$, there exists $x^*\geq 0$ such that 
\begin{gather} \label{x*_cdt1}
Ax^* F'(x^*)+AF(x^*)-v'(y)=0
\end{gather}
and 
\begin{gather} \label{x*cdt2}
\kappa_A(y)+Ax^* F(x^*)-x^* v'(y)=0, 
\end{gather}
where the equality in (\ref{x*_cdt1}) comes from the first-order condition of optimality of the expression 
$AxF(x)-xv'(y)$. But with $v'(y)=\frac{\kappa_A(y)}{G\bigl(\frac{\kappa_A(y)}{A}\bigr)}
+AF\bigl(G\bigl(\frac{\kappa_A(y)}{A}\bigr)\big)$, it can be checked that $x^*=G\bigl(\frac{\kappa_A(y)}{A}\bigr)$ 
satisfies both (\ref{x*_cdt1}) and (\ref{x*cdt2}). The boundary condition $v(0)=0$ is a consequence of the 
expression of $v(y)$ and the continuity of $v(y)$ at $y=0$. 
\end{proof}

\begin{proof}[\textbf{Proof of Theorem \ref{thmvarify}}]
We know that when $t\leq\tau$, 
\begin{gather*} 
\intl{Y^*_t}{y}\frac{1}{G\bigl(\frac{\kappa_A(u)}{A}\bigr)}\,du=t, 
\end{gather*}
from which it follows that 
\begin{gather*}
\xi^*_t=-\frac{dY^*_t}{dt}=G\biggl(\frac{\kappa_A(Y^*_t)}{A}\biggr), \qquad t\leq\tau. 
\end{gather*}
On the other hand, when $t>\tau$, $Y^*_t=0$. Hence, 
\begin{gather*}
\xi^*_t=0=G\biggl(\frac{\kappa_A(Y^*_t)}{A}\biggr), \qquad t>\tau. 
\end{gather*}

We next prove that $Y^*\in\mathcal{A}_D(y)$. It is clear that $Y^*$ is deterministic and absolutely 
continuous. The non-negativity of $G$ implies that 
$Y^*$ is non-increasing. It remains to show that if $\mu<0$, then $\int_0^\infty Y^*_t\, dt<\infty$; 
and if $\mu=0$, then $\int_0^\infty \bigl(Y^*_t\bigr)^2\, dt<\infty$. However, with reference to Lemma 
\ref{corr2}, it is enough to check that
\[\int_0^\infty \kappa_A\bigl(Y^*_t\bigr)\, dt=\int_0^\tau \kappa_A\bigl(Y^*_t\bigr)\, dt<\infty.\]
By a change of variable, we have that 
\[\int_0^\tau \kappa_A\bigl(Y^*_t\bigr)\, dt=
\int_y^0 -\frac{\kappa_A\bigl(Y^*_t\bigr)}{G\Bigl(\frac{\kappa_A(Y^*_t)}{A}\Bigr)}\,dY^*_t<\infty,\]
where the finiteness is du to the continuity of the integrand on the compact interval $[0,y]$, which is 
implied by Lemma \ref{lem_x/G(x)}. 

With reference to (\ref{x*_cdt1}) and (\ref{x*cdt2}), the function $v$ in (\ref{v(y)}) satisfies 
\begin{gather} \label{ineqHJB}
\kappa_A(y)+A\xi F(\xi)-\xi v'(y)\geq 0, \qquad\text{ for all }\xi\geq 0, 
\end{gather}
and equality holds only when $\xi=G\bigl(\frac{\kappa_A(y)}{A}\bigr)$. Let $Y\in\mathcal{A}_D(y)$. 
Observe that 
\[v(Y_T)=v(y)-\int^T_0v'(Y_t)\xi_t\, dt.\]
Taking $T$ to infinity and using the boundary condition $v(0)=0$, it follows that
\[v(y)=\int^\infty_0v'(Y_t)\xi_t\, dt.\]
Then by (\ref{ineqHJB}) we have 
\begin{gather}
v(y)\leq\int_0^\infty\Big(\kappa_A(Y_t)+A\xi_tF(\xi_t)\Big)\,dt. \label{vineq}
\end{gather}
Now consider the strategy $Y^*$ in (\ref{opt_strategy}), which has a speed process $\xi^*$ satisfying 
$\xi^*_t=G\bigl(\frac{\kappa_A(Y^*_t)}{A}\bigr)$, for all $t\geq 0$. Then, 
\begin{gather*} 
\kappa_A\bigl(Y^*_t\bigr)+A\xi^*_t F\bigl(\xi^*_t\bigr)-\xi^*_t v'\bigl(Y^*_t\bigr)=0, \qquad t\geq 0, 
\end{gather*}
hence, 
\begin{gather*}
v(y)=\int_0^\infty\Big(\kappa_A\bigl(Y^*_t\bigr)+A\xi_tF\bigl(\xi^*_t\bigr)\Big)\,dt. 
\end{gather*}
This together with (\ref{vineq}) implies that $V(y)=v(y)$, for all $y\in[0,\bar{\delta}_A)$. Therefore, with reference to the analysis after equation (\ref{reducedtoAD}), we get that $Y^*$ 
is the unique optimal strategy to problem (\ref{originalprob}). 
\end{proof}

\begin{proof}[\textbf{Proof of Proposition \ref{proptau}}]
\mbox{}
\begin{itemize}
\item[(i)]Suppose $\mu<0$ and  let $p<1$ be such that $\lim_{x\rightarrow 0}x^pF'(x)=K$, with $K$ being some strictly positive constant. 
Write $u=x^2F'(x)$. Then we have 
\[\frac{u^{\frac{1}{2-p}}}{G(u)}=\bigl(x^pF'(x)\bigr)^{\frac{1}{2-p}}.\]
By letting $x$ tend to $0$, so $u$ tends to $0$ as well, it follows that 
\begin{gather} \label{u^(1/2-p)/G(u)}
\lim_{u\rightarrow 0}\frac{u^{\frac{1}{2-p}}}{G(u)}=K^{\frac{1}{2-p}}.
\end{gather}
Lemma \ref{propkappaprop} (iv) together with (\ref{u^(1/2-p)/G(u)}) gives 
\begin{gather*}
\lim_{x\rightarrow 0}\frac{x^{\frac{1}{2-p}}}{G\bigl(\frac{\kappa_A(x)}{A}\bigr)}=K', 
\end{gather*}
for some other constant $K'>0$. 
Therefore, there exist strictly positive constants $K_1$, $K_2$ and $\bar{x}$ such that 
for all $x\in(0,\bar{x})$, 
\[\frac{K_1}{x^{\frac{1}{2-p}}}<\frac{1}{G\bigl(\frac{\kappa_A(x)}{A}\bigr)}<\frac{K_2}{x^{\frac{1}{2-p}}}.\]
Integrating and taking limit on each term gives 
\[\lim_{x\rightarrow 0}\int_x^{\bar{x}}\frac{K_1}{u^{\frac{1}{2-p}}}\,du
\leq\lim_{x\rightarrow 0}\int_x^{\bar{x}}\frac{1}{G\bigl(\frac{\kappa_A(u)}{A}\bigr)}\,du
\leq\lim_{x\rightarrow 0}\int_x^{\bar{x}}\frac{K_2}{u^{\frac{1}{2-p}}}\,du.\]
Observe that $p<1$ implies $\frac{1}{2-p}<1$, and therefore 
$\int_0^{\bar{x}}\frac{1}{u^{\frac{1}{2-p}}}\,du<\infty$. Hence, 
\begin{gather*}
\lim_{x\rightarrow 0}\int_x^{\bar{x}}\frac{1}{G\bigl(\frac{\kappa_A(u)}{A}\bigr)}\,du<\infty. 
\end{gather*}
Then the required result follows from (\ref{tau}) and the fact that 
$\int_{\bar{x}}^y\frac{1}{G\bigl(\frac{\kappa_A(u)}{A}\bigr)}\,du<\infty$, 
if the initial stock position $y>\bar{x}$.

\item[(ii)]Suppose $\mu=0$. Observe that (\ref{u^(1/2-p)/G(u)}) implies 
\[\lim_{x\rightarrow 0}\frac{x^{\frac{2}{2-p}}}{G(x^2)}=C,\]
for some constant $C>0$.  
Combining this with Lemma \ref{propkappaprop} (iii), we obtain 
\begin{gather*}
\lim_{x\rightarrow 0}\frac{x^{\frac{2}{2-p}}}{G\bigl(\frac{\kappa_A(x)}{A}\bigr)}=C', 
\end{gather*}
for some other constant $C'>0$. 
Then there exist strictly positive constants $C_1$, $C_2$ and $\bar{x}$ such that 
for all $x\in(0,\bar{x})$, 
\[\frac{C_1}{x^{\frac{2}{2-p}}}<\frac{1}{G\bigl(\frac{\kappa_A(x)}{A}\bigr)}<\frac{C_2}{x^{\frac{2}{2-p}}}.\]
Therefore, 
\[\lim_{x\rightarrow 0}\int_x^{\bar{x}}\frac{C_1}{u^{\frac{2}{2-p}}}\,du
\leq\lim_{x\rightarrow 0}\int_x^{\bar{x}}\frac{1}{G\bigl(\frac{\kappa_A(u)}{A}\bigr)}\,du
\leq\lim_{x\rightarrow 0}\int_x^{\bar{x}}\frac{C_2}{u^{\frac{2}{2-p}}}\,du.\]
If $p<0$, then $\frac{2}{2-p}<1$. Hence $\tau<\infty$ is obtained by the same argument as in (i) of this proof. 
If $p\in[0,1)$, then $\frac{2}{2-p}\geq 1$. It follows that 
$\int_0^{\bar{x}}\frac{1}{G\bigl(\frac{\kappa_A(u)}{A}\bigr)}\,du=\infty$, and therefore $\tau=\infty$. 
\end{itemize}
\end{proof}

\begin{proof}[\textbf{Proof of Proposition \ref{thmnuL}}]
We show that $\hat{L}$ given by 
\begin{align}
\hat{L}_t
&=\tilde{m}t+\tilde{\sigma}\tilde{W}_t
+\int_0^t\int_{\ktR}\bigl(\mathrm{e}^z-1\bigr)
\,\bigl(\tilde{N}(dt,dz)-\tilde{\nu}(dz)dt\bigr), \qquad t\geq 0,  \label{Lhat_barN-barnu}
\end{align}
is a L\'evy process. Define a random measure 
$\hat{N}:\Omega\times\mathcal{B}\bigl([0,\infty)\bigr)\otimes\mathcal{B}(\ktR)\rightarrow\mathbb{Z}_+$ 
and a measure $\hat{\nu}:\mathcal{B}(\ktR)\rightarrow\mathbb{Z}_+$ to be such that if 
$B\in\mathcal{B}(\ktR)$ and $B\cap(-1,\infty)\neq \emptyset$, then 
\begin{align}
\hat{N}(\omega,A,B)&=\tilde{N}\Bigl(A\,,\,\ln\bigl(B\cap(-1,\infty)+1\bigr)\Bigr)(\omega), \cr
\hat{\nu}(B)&=\tilde{\nu}\Bigl(\ln\bigl(B\cap(-1,\infty)+1\bigr)\Bigr); \label{def_nuhat}
\end{align}
otherwise, they are both equal 0, where $\mathbb{Z}_+$ is the set of all positive integers 
and $\ln(B\cap(-1,\infty)+1)=\{\ln(x+1)\,|\,x\in B\cap(-1,\infty)\}$ (\,we have for all 
$A\in\mathcal{B}([0,\infty))$ and $\omega\in\Omega$, $\tilde{N}(A,\{0\})(\omega)=\tilde{\nu}(\{0\})=0$\,). 
Write $\hat{N}(\cdot,\cdot)=\hat{N}(\omega,\cdot,\cdot)$. Then by writing $x=\mathrm{e}^z-1$, 
it follows from (\ref{Lhat_barN-barnu}) that 
\begin{gather} \label{Lhat_Nhat-nuhat}
\hat{L}_t=\tilde{m}t+\tilde{\sigma}\tilde{W}_t
+\int_0^t\int_{\ktR}x\,\bigl(\hat{N}(dt,dx)-\hat{\nu}(dx)dt\bigr), \qquad t\geq 0. 
\end{gather}
With reference to \citet{Kall} Corollary 15.7, to prove $\hat{L}$ is a L\'evy process, 
it suffices to show that for any $B\in\mathcal{B}(\ktR)$, $\bigl(\hat{N}(t,B)\bigr)_{t\geq 0}$ 
is a Poisson process with intensity $\hat{\nu}(B)$ satisfying 
\begin{gather} \label{x^2wedge1}
\int_{\ktR} \bigl(x^2\wedge 1\bigr)\,\hat{\nu}(dx)<\infty. 
\end{gather}
But from the definition of $\hat{N}$, it is clear that $\bigl(\hat{N}(t,B)\bigr)_{t\geq 0}$ is 
a Poisson process. Observe that 
\begin{gather*}
\ktE[\hat{N}(t,B)]=\ktE\bigl[\tilde{N}\bigl(t,\ln(B\cap(-1,\infty)+1)\bigr)\bigr]=t\tilde{\nu}\bigl(\ln(B\cap(-1,\infty)+1)\bigr)=t\hat{\nu}(B), 
\end{gather*}
which proves that $\hat{\nu}(B)$ is the intensity of $\bigl(\hat{N}(t,B)\bigr)_{t\geq 0}$. 
From the Taylor expansion of $(\mathrm{e}^z-1)^2$, it can be shown that there exist constants 
$\bar{z}>0$ and $C>0$ such that for all $z\in(-\bar{z},\bar{z})$, 
\[\bigl(\mathrm{e}^z-1\bigr)^2\leq Cz^2.\]
For $\epsilon\in(0,1)$, consider interval $\mathcal{S}=\bigl(\,\ln(1-\epsilon)\,,\,\ln(\epsilon+1)\,\bigr)$. 
Then using (\ref{def_nuhat}) we calculate that for $\epsilon$ close enough to 0 so that 
$\mathcal{S}\subseteq(-\bar{z},\bar{z})$, we have 
\begin{gather} \label{epsilonx^2nudx}
\int_{(-\epsilon,\epsilon)} x^2\,\hat{\nu}(dx)
=\int_{\mathcal{S}}\bigl(\mathrm{e}^z-1\bigr)^2\,\tilde{\nu}(dz)
\leq C\int_{\mathcal{S}}z^2\,\tilde{\nu}(dz)
\leq C\int_{(-\bar{z},\bar{z})}z^2\,\tilde{\nu}(dz)<\infty, 
\end{gather}
where the finiteness follows since  $\tilde{\nu}$ is a L\'evy measure. 
Again by (\ref{def_nuhat}), we obtain 
\begin{gather*}
\int_{\ktR\setminus(-\epsilon\,,\,\epsilon)}\,\hat{\nu}(dx)
=\int_{\ktR\setminus\mathcal{S}}\,\tilde{\nu}(dz)<\infty, 
\end{gather*}
where the finiteness again follows since $\tilde{\nu}$ is a L\'evy measure . This implies that 
$\hat{\nu}\bigl(\ktR\setminus(-1,1)\bigr)<\infty$ and $\hat{\nu}\bigl((-1,-\epsilon]\cup[\epsilon,1)\bigr)<\infty$. 
Since $x^2$ is bounded on $(-1,-\epsilon]\cup[\epsilon,1)$, together with (\ref{epsilonx^2nudx}), we get 
\begin{gather*}
\int_{(-1,1)} x^2\,\hat{\nu}(dx)<\infty. 
\end{gather*}
Combining this with $\hat{\nu}\bigl(\ktR\setminus(-1,1)\bigr)<\infty$, we get (\ref{x^2wedge1}). We therefore conclude that $\hat{N}$ and $\hat{\nu}$ are Poisson random measure and 
L\'evy measure associated with the L\'evy process $\hat{L}$, respectively. Moreover, we calculate from (\ref{def_nuhat}) 
that for $x>-1$ and $x\neq 0$, 
\begin{gather*}
\hat{\nu}(dx)=\tilde{\nu}\bigl(\,d\bigl(\ln(x+1)\bigr)\,\bigr)
=\tilde{f}\bigl(\ln(x+1)\bigr)\,d\bigl(\ln(x+1)\bigr)=\frac{1}{x+1}\tilde{f}\bigl(\ln(x+1)\bigr)dx. 
\end{gather*}

The relation $L=\tilde{s}\hat{L}$ shows that 
$L$ is also a L\'evy process. The expression of $L$ in (\ref{candidate}) shows the adaptedness. 
Now we check Assumption \ref{assum1} is satisfied by $L$, but it suffices to check for $\hat{L}$. 
According to Assumption \ref{assum3}, 
we know $\int_{|z|\geq 1}\mathrm{e}^{2z}\,\tilde{\nu}(dz)<\infty$, and since for any $\epsilon>0$, 
$\tilde{\nu}\bigl(\ktR\setminus(-\epsilon,\epsilon)\bigr)<\infty$, it follows that on $[\ln 2,\infty)$, 
$\e^{2z}$ and $\e^z$ are both $\tilde{\nu}$-integrable and $\tilde{\nu}\bigl([\ln 2,\infty)\bigr)<\infty$. 
Therefore, 
\[\int_{|x|\geq 1}x^2\,\hat{\nu}(dx)=\int_{[\ln 2,\infty)}\bigl(\mathrm{e}^z-1\bigr)^2\,\tilde{\nu}(dz)<\infty,\]
which implies that $\hat{L}_1$ has finite second moment \citep[see e.g.][Theorem 3.8]{Kypr}. Observe that when $u\leq 0$, we have 
\[\exp\bigl(u(\mathrm{e}^z-1)\bigr)\leq 1,\quad\text{ for all }z\geq 0.\] 
Hence, 
\begin{gather}
\int_{|x|\geq 1}\mathrm{e}^{ux}\,\hat{\nu}(dx)
=\int_{[\ln 2,\infty)}\exp\bigl(u(\mathrm{e}^z-1)\bigr)\,\tilde{\nu}(dz)<\infty, \label{forremarkinsec5}
\end{gather}
from which it follows that $\ktE[\e^{u\hat{L}_1}]<\infty$, for all $u\leq 0$. 
\end{proof}

\begin{proof}[\textbf{Proof of Theorem \ref{Thm_appro}}]
This is a direct consequence of Theorem \ref{thmvarify}. 
\end{proof}

\begin{proof}[\textbf{Proof of Proposition \ref{prop_lowerbd}}]
For $u\geq 0$, we calculate that 
\begin{align} \label{proof_prop_lowerbd1}
&\int_{-1}^0 \Bigl(\mathrm{e}^{-\tilde{A}ux}-1+\tilde{A}ux\Bigr)\,\hat{\nu}(dx)\notag\\
=\,\,&\int_{-1}^0 \Bigl(\mathrm{e}^{-\tilde{A}ux}-1+\tilde{A}ux\Bigr)\frac{-1}{\eta\ln(x+1)}(x+1)^{C+D-1}\,dx\notag\\
\geq\,\,&\frac{\e}{\eta}\int_{-1}^0 \Bigl(\mathrm{e}^{-\tilde{A}ux}-1+\tilde{A}ux\Bigr)(x+1)^{C+D}\,dx\notag\\
=\,\,&\frac{\e}{\eta}\int_{0}^1 \Bigl(\mathrm{e}^{-\tilde{A}u(x-1)}x^{C+D}\Bigr)\,dx
+\frac{\e}{\eta}\int_{0}^1 \Bigl(\tilde{A}ux^{C+D+1}\Bigr)\,dx
+\frac{\e}{\eta}\int_{0}^1 \Bigl(-(1+\tilde{A}u)x^{C+D}\Bigr)\,dx\notag\\
\end{align}
where the first inequality is due to $\frac{-1}{(x+1)\ln(x+1)}\geq \e$, 
for all $-1<x<0$, since $(x+1)\ln(x+1)$ is convex with minimum value $-\e^{-1}$. Observe that 
\begin{align}\label{proof_prop_lowerbd2}
&\int_{0}^1 \Bigl(\mathrm{e}^{-\tilde{A}u(x-1)}x^{C+D}\Bigr)\,dx\cr
\geq\,\,&\mathrm{e}^{\tilde{A}u}\int_{0}^{\frac{1}{\tilde{A}u}\wedge 1}
\Bigl(\bigl(-\tilde{A}ux+1\bigr)x^{C+D}\Bigr)\,dx\cr
=\,\,&-\frac{\tilde{A}u\e^{\tilde{A}u}}{C+D+2}\biggl(\frac{1}{\tilde{A}u}\wedge 1\biggr)^{C+D+2}
+\frac{\e^{\tilde{A}u}}{C+D+1}\biggl(\frac{1}{\tilde{A}u}\wedge 1\biggr)^{C+D+1}
\end{align}
and 
\begin{gather}\label{proof_prop_lowerbd3}
\int_{0}^1 \Bigl(\tilde{A}ux^{C+D+1}\Bigr)\,dx+\int_{0}^{1} \Bigl(-(1+\tilde{A}u)x^{C+D}\Bigr)\,dx
=\frac{\tilde{A}u}{C+D+2}-\frac{1+\tilde{A}u}{C+D+1}, 
\end{gather}
where we have $C+D>0$ and the inequality is because that $\e^{-\tilde{A}ux}\geq-\tilde{A}ux+1$ 
on interval $\bigl[0,\frac{1}{\tilde{A}u}\wedge 1\bigr]$. Therefore, the required result follows from (\ref{proof_prop_lowerbd1})-(\ref{proof_prop_lowerbd3}) and the 
expression of $\hat{\kappa}_{\tilde{A}}^{VG}$ in (\ref{kappaofVG}) as well as the fact that 
$\mathrm{e}^{-\tilde{A}ux}-1+\tilde{A}ux$ and $\hat{\nu}$ are positive for all $u\geq 0$ and $x\in \ktR$. 
\end{proof}


\bibliographystyle{apalike} 
\bibliography{references}

\end{document}